\documentclass[10pt,journal,compsoc]{IEEEtran}
\usepackage{amssymb}
\usepackage{array}
\usepackage{float}
\usepackage{stfloats}
\usepackage{color}

\usepackage{booktabs}
\usepackage{graphicx}
\usepackage{balance}  
\usepackage{enumerate}
\usepackage{tabularx}
\usepackage{multirow}
\usepackage{makecell}
\usepackage{amsthm,amsmath}
\usepackage{mathrsfs}
\makeatletter  
\newif\if@restonecol  
\makeatother

\usepackage[linesnumbered,ruled,vlined]{algorithm2e}
\usepackage{algpseudocode}  

\usepackage{subfigure}

\newtheorem{example}{Example}

\newtheorem{theorem}{Theorem}
\usepackage{bm}
\definecolor{gray}{RGB}{190,190,190}

\newcommand{\yl}[1]{{\color{black}{#1}}}
\newcommand{\cy}[1]{{\color{black}{#1}}}

\newcommand{\zyl}[1]{{\color{black}{#1}}}
\newcommand{\cs}[1]{{\color{black}{#1}}}
\newcommand{\zzzyl}[1]{{\color{black}{#1}}}

\ifCLASSOPTIONcompsoc
  \usepackage[nocompress]{cite}
\else
  \usepackage{cite}
\fi
\ifCLASSINFOpdf
\else
\fi
\begin{document}
%
\title{LSketch: A Label-Enabled Graph Stream Sketch Toward Time-Sensitive Queries}
%
%
%
%

\author{Yiling~Zeng,
        Chunyao~Song,
        Yuhan~Li,
        and~Tingjian~Ge 
\IEEEcompsocitemizethanks{\IEEEcompsocthanksitem Y. Zeng, C. Song, and Y. Li are with the Tianjin Key Laboratory of Network and Data Security Technology, College of
Computer Science, Nankai University, Tianjin 300350, China. E-mail: \{yilingzeng, yuhanli\}@mail.nankai.edu.cn, chunyao.song@nankai.edu.cn    

\IEEEcompsocthanksitem T. Ge is with the Department of Computer Science, University of Massachusetts Lowell. E-mail: ge@cs.uml.edu.}

}

\IEEEtitleabstractindextext{%
\begin{abstract}
Graph streams represent data interactions in real applications. The mining of graph streams plays an important role in network security, social network analysis, and traffic control, among others. However, the sheer volume and high dynamics cause great challenges for efficient storage and subsequent query analysis on them. Current studies apply sketches to summarize graph streams.
We propose LSketch that works for heterogeneous graph streams, which effectively preserves the label information carried by the streams in real scenes, thereby enriching the expressive ability of sketches.
In addition, as graph streams continue to evolve over time, edges too old may lose their practical significance. Therefore, we introduce the sliding window model into LSketch to eliminate the expired edges automatically.
LSketch uses sub-linear storage space and can support structure based queries and time-sensitive queries with high accuracy.
We perform extensive experiments over four real datasets, demonstrating the superiority of the proposed method over state-of-the-art methods, in aspects of query accuracy and time efficiency.
\end{abstract}

}

\maketitle

\IEEEdisplaynontitleabstractindextext

%
\IEEEpeerreviewmaketitle

\IEEEraisesectionheading{\section{Introduction}\label{sec:introduction}}

%
%
%
%
\IEEEPARstart{I}{n} this big data era, data is generated in real-world networks at ever-increasing high rates.
\yl{As of 1st quarter 2022, Facebook has 2.91 billion monthly active users and 79\% of monthly users are active daily \footnote{https://www.statista.com/statistics/264810/number-of-monthly-active-facebook-users-worldwide/}. These active users generate huge amounts of data on the platform.}
\yl{\cy{On this basis, }
a wide variety of real-world applications \cy{which can be modeled with graphs} 
have been extensively investigated \cite{aggarwal2011introduction,aggarwal2010managing,tsalouchidou2018scalable}. However, most of the past \cy{studies concentrated more }
on static graphs. In more recent years, streams on large-scale graph infrastructures have been proposed \cy{due to different }
scenario requirements of \cy{variant }
applications \cite{cormode2005space,feigenbaum2005graph,ganguly2006estimating,sarma2011estimating,aggarwal2010dense}.
}
\cy{They} are used to model graphs that are sequentially updated in the form of edges as an extension of static graphs, which simulate the continuous evolution of the networks in real scenarios.
Similar to static graph mining, the analytics of large-scale graph streams is of great practical significance.
For social networks, exploring the connections between nodes helps predict a user's potential friends or detect the source of fake news, for example. For transportation networks, congested road section prediction and route planning also benefit from such analyses. 

It is usually impossible to accurately compute the frequency of edges or nodes over enormous graph streams due to its sheer amount and excessive change rate. 
\yl{Accordingly, approximate queries \cy{serve to solve the problem }
and are receiving increasing attention.
In \cy{past studies, }
a series of 
data structures \cy{which work for approximate processing} have been proposed, including Bloom filters \cite{aggarwal2007data} and its variances \cite{aggarwal2011introduction,aggarwal2010managing,aggarwal2010dense}, and sketches \cite{alon1999space,buriol2006counting,bar2002reductions,estan2003new,roy2016augmented,yang2017pyramid}.
}
\yl{
\cy{Such} data structures, which can markedly reduce the storage space and get extremely fast answering speed with a limited loss of accuracy, are effective methods for big data applications\cy{, especially in streaming scenarios}.
}
\yl{At present, they are widely used for finding top-k items \cite{chang2003finding,jin2008sliding}\cy{;} finding heavy-hitters \cite{mirylenka2015conditional,ben2016heavy}\cy{;} approximate weight estimation \cite{manku2002approximate,zhou2019generalized}\cy{;} and triangle counting \cite{pavan2013counting,shin2017wrs}.
}

In addition, rich information is carried by nodes and edges in heterogeneous graphs. Ignoring such information can result in serious knowledge loss.
For example, in social networks, users can be treated as nodes, and edges represent the communications/interactions between them. Users can be assigned to different communities according to their interests and thus have various personal tags---labels on nodes. The communication between users can also be classified according to the communication intensity, such as frequent, medium, and infrequent ones. 
\yl{However, \cy{most }current sketch techniques 
only consider nodes and edges without labels, while the data in real applications are often labeled on nodes and edges. The few existing studies suffer from poor query accuracy and cannot meet the needs of realistic scenarios.
}

\yl{Furthermore, most current sketch techniques do not consider the influence of insertion rates, nor differentiate items arriving at different timestamps when maintaining the sketches.
On one hand, the update frequency of graph streams represents the active status of the network, so the sketch should be able to carry time related information.
On the other hand, items arriving at different timestamps are of different importance in practical applications. The closer to the current time, the more valuable an item is. 

In summary, both labels and timestamps are important features for real-world graph streams, thus we should construct a time-sensitive sketch that can work well for heterogeneous graph streams.
}

\subsection{Related Work} \label{relatedwork}
Due to the rapid growth of data, the graph stream model has been explored for data analytics and query processing. 
McGregor gives an excellent survey~\cite{mcgregor2014graph} of mostly theoretical work on graph streams.
The problem of synopsis construction has been widely studied for {\em data streams} in general. The early graph stream summary techniques apply linear projections of the data items with multiple hash functions into lower dimensional spaces to store each item independently, but ignore the connections between items. 
\yl{However, such sketches are not applicable to the case of graph data. Instead of considering the potential relevance of edges in the graph stream, they model\cy{ed} edges as a series of independent items, which \cy{cannot preserve }
the underlying structure of the graph stream.} Therefore, they can only support edge queries rather than any \cy{topology }
based queries. This line of work includes \yl{Count Sketch \cite{charikar2004finding}}, CountMin Sketch \cite{Cormode2005AnID}, gSketch \cite{Zhao2011gSketchOQ}, AMS \cite{alon1999space}, 
Ada-sketch \cite{shrivastava2016time} and so on.

To overcome \yl{the shortcomings of these early studies, graph sketches are proposed to improve the above classical sketches for graph streams}.
\zyl{TCM \cite{Tang2016GraphSS}, gMatrix \cite{khan2016query} and LGS \cite{Song2019LabeledGS}} use three-dimensional sketches to support graph structure based queries. They apply a hash function to map the vertex set of the graph into an integer within the size of the matrix width, thus the item (an edge) can be located into the matrix using hash values of its two endpoints. 
\cy{Although they are able to support graph topology based queries, they still suffer from poor accuracy. To address this problem, }
GSS \cite{Gou2019FastAA, gou2022graph} introduces a series of techniques, such as square hashing and multiple rooms, to take care of the uneven distribution of the vertex degrees, thereby increasing the capacity of the matrix and achieving the highest query accuracy thus far.

However, \zyl{most of the methods} introduced above are studied over {\em homogeneous} graph streams. By contrast, data in real applications often carry a lot of labeled information. 
Furthermore, those sketches are not time-sensitive and cannot tell the active status of the network.
The edges that are too old cannot be removed in time, which will affect subsequent queries.
In recent years, Hung et al. \cite{hung2008finding} study the problem of identifying items with heavy weights in the sliding window of a weighted data stream.
ECM-sketch \cite{papapetrou2012sketch} allows effective summarization of streaming data over both time-based and count-based sliding windows to answer potentially complex queries with probabilistic accuracy guarantees.
SBG-Sketch \cite{hassan2018sbg} summarizes labeled-graph streams and automatically balances sketch load with unpredictable and highly imbalanced edge-label frequencies. 
Extending TCM, LGS \cite{Song2019LabeledGS} preserves the timestamps of items by introducing the sliding window model and automatically handles edge expiration. Moreover, it stores vertex labels and edge labels efficiently, but its query accuracy is not high enough for practical usage. 

\zyl{In summary, we choose GSS, the sketch with the highest accuracy working for homogeneous graph streams, and LGS, the sketch with the highest accuracy working for heterogeneous graph streams, as the competitive methods to demonstrate the superiority of our sketch in retaining labeled and temporal information.}

\subsection{\cy{Our }\yl{Contributions}}
\cy{To solve the problems mentioned above}, it is of great practical significance to construct a graph sketch with 
sliding window\cy{s} 
over heterogeneous graph streams, which can support \cy{graph topology related queries }
and achieve high \cy{efficiency and } accuracy \cy{as well}. 



\yl{Our contributions are summarized as follows:}

1) We propose LSketch for heterogeneous graph streams, which stores vertex labels and edge labels using the idea of Storage Blocks Division and Dual Counters without occupying much extra storage space. 

2) We introduce a sliding window model consisting of $k$ subwindows into an LSketch. By tracking the start time of the latest subwindow, we can automatically handle edge expirations in an efficient way, \zyl{so that the sketch effectively stores the temporal information of the items.}

3) \yl{Specifically, we design two strategies for the encoding of vertex labels, \cy{named }uniform blocking and skewed blocking\cy{. Particularly, }
the skewed blocking strategy can adapt to the extreme unbalanced distribution of vertex labels \cy{and achieve similar high efficiency}.}

4) We show some sample implementations of the structure based queries on LSketch, and conduct extensive experiments on four real-world datasets. The results demonstrate that \zyl{under the theoretical guarantees,} LSketch outperforms LGS in \cs{terms of }query accuracy \cs{by 1 to 3 orders of magnitude, }\zyl{and \cs{still retain comparable }
query efficiency with it}. Furthermore, it supports more types of queries than GSS.

\section{Preliminaries} \label{sec:preliminaries}

\subsection{Heterogeneous Graph Stream} \label{sec:graphstreams}

\begin{figure}[htbp]
\vspace{-0.3em}
\setlength{\abovecaptionskip}{0.1cm}
\centering
\includegraphics[width=1.9in]{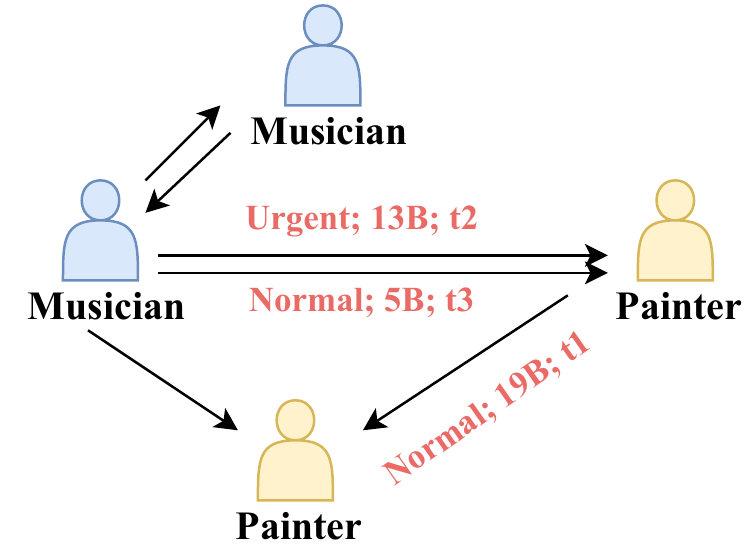} 
\caption{A heterogeneous graph stream in social networks} 
\label{fig:socialnetworks}
\vspace{-0.7em}
\end{figure}

A heterogeneous graph stream is a real-time updated graph that is constituted by edges arriving in the form of item $e=(A,B;l_A,l_B,l_e;w;t)$, which means that the edge $(A, B)$ comes at time $t$ with weight $w$. In addition, $l_A$ and $l_B$ denote the label attributes of nodes A and B, respectively, and $l_e$ denotes the edge label.
As time goes on, items may appear multiple times with different weights. We take Figure \ref{fig:socialnetworks} as an example. In social networks, the interaction between users can be modeled as an edge, so the evolution among users' different interactions forms a graph stream, in which the user's community can act as its vertex label, such as musician or painter. Also, the edge label can stand for the urgency level selected by the user when sending the message. It implies that it is possible that edges with different labels appear between the same pair of nodes $A$ and $B$, with labels $l_A$ and $l_B$, respectively. 
A message corresponds to an item/edge in the graph stream, with its weight being the length of the message. 
\zyl{The symbols frequently used in this paper are listed in Table \ref{tableG}. }
\begin{table}  
\caption{\zyl{List of symbols used in the paper} }
\centering
\begin{tabular}{ll}
\hline  
Symbol & Definition \\  
\hline  
$e $  &   The arriving item of the graph stream\\
$A/B $  &  The starting/ending node of item $e$  \\ 
$l_A/l_B$   &  The vertex label of node $A/B$ \\ 
$l_e$   &  The edge label of item $e$ \\ 
$w$ &  The weight of item $e$  \\
$t$ & The timestamp of item $e$ \\ 
$d$ & Width of the storage matrix  \\
$b$ & Width of the storage blocks  \\
$m$ & The index of block of a certain vertex label \\
$H(v)$ & The calculated hash value of node $v$  \\
$f(v)$ & The calculated fingerprint of node $v$  \\
$s(v)$ & The calculated address of node $v$  \\
$r$ & Length of the sampling sequence \\
$s$ & Length of the sampled cells \\
$l_i(v)$ & The address candidate list of node $v$ \\
$C$ & Counter $C$ works for queries without label restrictions \\
$P$ & Counter $P$ works for queries with label restrictions  \\
$c$ & The length of a predefined list of prime numbers\\
$W$ & Time units of the whole sliding window  \\
$W_s$ & Time units of one subwindow  \\
$k$ & The size of subwindows \\

\hline  
\end{tabular} 
\label{tableG}
\end{table}

Given the stream, we can define its underlying structural graph $G=(V, E)$, which is a dynamic directed graph that is continuously updated as each edge item arrives. Here, $V$ is the vertex set, and $E$ is the edge set.
For any vertex $v_i\in V$, there is an attached {\em vertex label} $l_V(v_i)$. Similarly, for any edge $e_i\in E$, there is an attached {\em edge label} $l_E(e_i)$. 
As mentioned earlier, edges may appear several times with either the same or different weights and labels in the stream; thus we allow multi-edges between the same pair of nodes.

Due to its sheer amount and high dynamicity, the storage and processing of graph streams are usually difficult to resolve. 
Traditional data structures, such as adjacency matrices, cannot be directly applied either.
Therefore, we need to construct a real-time graph summarization model, the graph sketch, to process graph streams with less storage space and faster response speed.

A graph sketch of the underlying structural graph $G=(V, E)$ is represented as $G_s=(V_s, E_s)$, where $V_s$ is the vertex set and $E_s$ is the edge set. 
Ideally, $V_s$ should be significantly smaller than $V$, and so is $E_s$ than $E$. 
In order to achieve the goal of reducing storage space and supporting structure based queries, a well-designed data structure is required to map edges in the graph stream to the sketch.

\subsection{\yl{GSS: A Novel Sketch Supporting Homogeneous Graph Streams}\cy{ with High Accuracy} } \label{sec:gss}
We first introduce the basic sketch \yl{GSS} \cite{gou2022graph} that can store homogeneous graph \yl{streams} in a much compressed way. 

\yl{GSS defines} the generated sketch as $G_s=(V_s, E_s)$, and a matrix of width $d$ forms the underlying data structure. 
When an item $e=(A,B;w)$ arrives in the data stream, \yl{it applies} a hash function $H(\cdot)$ and \yl{maps} the nodes into a value range $[0,D)$ (the value of D will affect the accuracy of the sketch). Then we get the hash values $H(A)$ and $H(B)$ corresponding to the two vertices, respectively.

\noindent
\textbf{Fingerprint Technique.} 
\begin{figure}[ht]
\centering
\setlength{\abovecaptionskip}{0.1cm} 
\setlength{\belowcaptionskip}{-0.2cm}   
\includegraphics[width=1.7in]{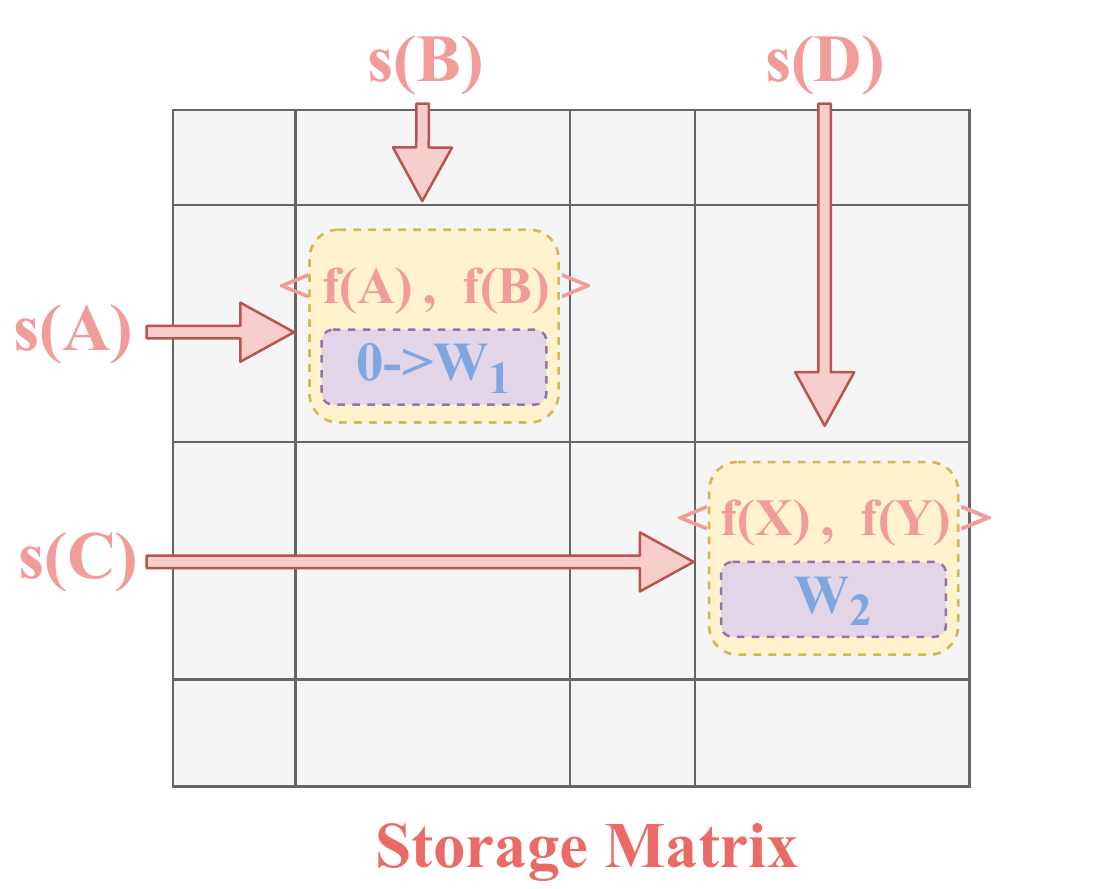} 
\caption{A running example of the Storage Matrix. Item $(A,B)$ inserts into an empty cell successfully, while item $(C,D)$ reaches an occupied cell, and the insertion fails.} 
\label{fig:storagematrix}
\end{figure}
\yl{GSS splits} the hash value $H(v)$ of vertex $v$ into an address $s(v) (0\leq s(v)\leq d)$ and a fingerprint $f(v)(0\leq f(v)\leq F)$, where $s(v)= \lfloor \frac{H(v)}{F} \rfloor$, $f(v)= H(v)\% F$, and $F$ denotes the size of the fingerprint (e.g., $F = 1024$ implies a 10-bit fingerprint). At the same time, \yl{it sets} $D$ to $d*F$, which is the largest hash range that a matrix of width $d$ can accurately express.
Through the address pair $(s(A),s(B))$, \yl{GSS} can locate the cell where the item should be stored in the matrix. Besides the weight $w$, \yl{it also stores} the fingerprint pair $(f(A),f(B))$ as a unique identifier to prevent the item from being overwritten due to hash conflicts in subsequent updates.

Figure \ref{fig:storagematrix} shows an example of the actual operations when items are processed.
\yl{If the cell is not empty or the two fingerprint pairs do not match, then the edge cannot be stored in the storage matrix. GSS adds another data structure called {\em Buffer} for such edges, as detailed below.}


\begin{figure}[htbp]
\centering
\setlength{\abovecaptionskip}{0.1cm} 
\setlength{\belowcaptionskip}{-0.2cm}   
\includegraphics[width=2.4in]{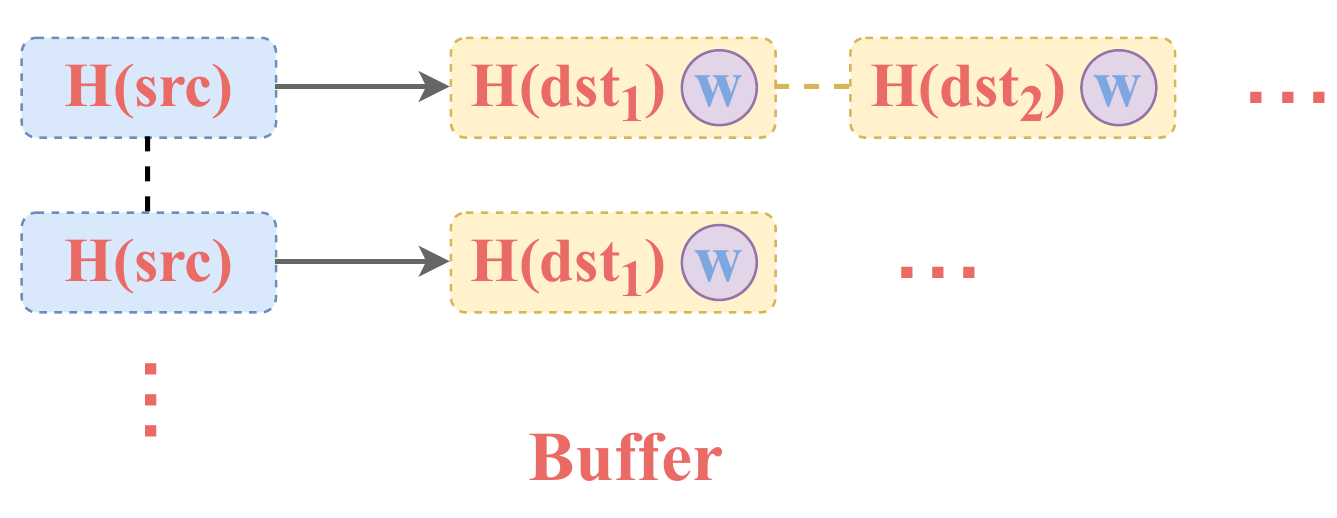} 
\caption{The structure of Buffer} 
\label{fig:buffer}
\end{figure}

\noindent
\textbf{Buffer.}  
Since the buffer \yl{is designed to store the left-over edges,} \yl{GSS simply uses} an adjacency list.
It is illustrated as in Figure \ref{fig:buffer}, in which the hash value $H(v)$ of a vertex $v$ serves as its identifier. In summary, \yl{GSS consists} of an optimized storage matrix and an adjacency list serves as the buffer. 


\noindent
\textbf{Twin Cells.} 
As mentioned \yl{before}, a matrix cell only stores edges with the same fingerprints. When a fingerprint mismatch occurs during insertion, the edge will be put in the \yl{buffer}.
In order to improve the insertion success rate of the storage matrix, \zyl{GSS introduces Twin Cells to double its capacity without increasing the storage of the matrix}.

\begin{figure}[htbp]
\centering
\setlength{\abovecaptionskip}{0.1cm} 
\setlength{\belowcaptionskip}{-0.2cm}   
\includegraphics[width=1.2in]{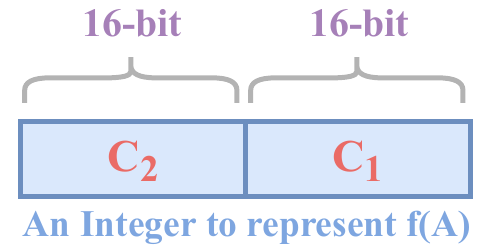} 
\caption{Illustrating twin cells in GSS} 
\label{fig:2room}
\end{figure}

Specifically, \yl{it divides} all the configurations (including a fingerprint pair and the corresponding \yl{weight}) in a cell into two twin segments, $C_1$ and $C_2$ (i.e., two children). For example, we \cy{can }use 16 bits to represent the fingerprints of $C_1$ and $C_2$, respectively, and concatenate them in one 32-bit word.
Therefore, each of the twin segments $C_1$ and $C_2$ is able to store a unique item, as illustrated in Figure \ref{fig:2room}.
When an edge locates a certain cell based on hashing and fingerprints, \yl{GSS checks} the status of the lower segment $C_1$ and the higher segment $C_2$ for the insertion in turns. Once $C_1$ is occupied and $C_2$ is available, \yl{it needs} to combine the lower segment $C_1$ with the newly inserted higher segment $C_2$ using bit operations, \yl{and then update the corresponding weights.} \zyl{In this way, \cs{GSS is }
able to reduce the memory accesses to finding the bucket without increasing storage, thus improving the efficiency of insertion and query of the entire sketch.}

\noindent
\textbf{Square Hashing And Sampling.} 
To further alleviate the matrix congestion caused by node degree skewness, \yl{GSS no longer insists} on mapping a vertex to a fixed row/column. 
\yl{Rather, it generates a series of mapped cells for each edge to make insertion attempts, and thus realizing decentralized distribution in the storage matrix by extending the mapping selections of vertices.
To achieve the two properties of \textbf{Independent} \cite{Cormode2005AnID, khan2016query} and \textbf{Reversible}, GSS accomplishes the generation of square hashing by means of linear congruence method \cite{l1999tables}.
}
\yl{GSS sets} $f(v)$ as the seed to generate a sequence of random values of length $r$, which can be expressed as 
\begin{equation}  \label{linear}
   \left\{
             \begin{array}{lr}
             l_1(v)=(T\times f(v)+I)\%M  \\
             l_i(v)=(T\times l_{i-1}(v)+I)\%M   \quad (2\leq i\leq r)\\
             \end{array}
    \right.
\end{equation}
where the multiplier $T$, the increment $I$, and the modulus $M$ are integer constants that specify the generator. 
Appropriate parameters can result in a sequence without duplicate values with a period much greater than $r$.
Based on the random value sequence, GSS arrives at the address candidate list as:
\begin{equation}  \label{linear2}
   \{s_i(v)|s_i(v)=(s(v)+l_i(v))\%d,1\leq i\leq r\}
\end{equation}

The length $r$ of the sequence is mainly set based on the skewness of datasets. It is time-consuming if all $r\times r$ cells are checked, while most edges do not need to traverse all of these in an ideal situation. Hence, GSS does not check all the mapped cells. Instead, it use the linear congruence method again to select $s$ buckets among them as sampled cells.
It calculates the sum of the fingerprints of the two endpoints of an edge to represent it, and use the summation as a seed to generate the following sequence:
\begin{equation}  
   \left\{
             \begin{array}{lr}
             Sp_1(e)=(T\times (f(A)+f(B))+I)\%M  \\
             Sp_i(e)=(T\times Sp_{i-1}(e)+I)\%M   \quad\quad  (2\leq i\leq s)\\
             \end{array}
    \right.
\end{equation}

Each value in sequence $Sp$ is then transformed into the address subscripts of the two endpoints using the same method as the fingerprint technique, thus

\begin{equation}  
\setlength{\abovedisplayskip}{3pt}
   \left\{
             \begin{array}{lr}
             A_i =  \lfloor \frac{Sp_i(e)}{r} \rfloor \% r    \\  
             B_i =  Sp_i(e)\%r   \quad\quad\quad\quad\quad\quad\quad\quad\quad (1\leq i\leq s) \\
             Sampled \, cell_i = (s_{A_i}(A),s_{B_i}(B))     \\  
             \end{array}
    \right.
\end{equation}
where $A_i$ denotes the selected subscript of the address candidate list of the starting vertex, while $B_i$ denotes the same of the ending vertex.

\section{The Construction Of LSketch} \label{sec:lsketch}
\yl{\cy{Although }GSS works well on homogeneous graph streams, \cy{it is not able to support heterogeneous graph stream based queries, which are more needed in real scenarios. }
Therefore, we propose to build on GSS to construct sketches that can adapt to heterogeneous graph streams.}
In this section, we will show how to build an \textbf{LSketch} (i.e., Label-enabled graph stream Sketch) step by step. 
We start by \yl{introducing how to store labeled information and how to apply sliding windows \cy{to strike an emphasis on preserving time sensitive information}.}
Then we will \cy{show pseudocode and running examples to better illustrate how LSketch works in practice, and complete the section by providing the improvement methods of LSketch to accommodate the needs of different types of datasets. }

\subsection{Labels Enabled \yl{Matrix}} \label{sec:block}   
As mentioned in Section \ref{sec:introduction}, the labels of nodes and edges play important roles in analyzing graph based applications. We now discuss how to accommodate labels on top of \yl{GSS.}
For a heterogeneous graph stream, let an edge be in the form of $e=(A,B;l_A,l_B,l_e;w)$, where $l_A$ and $l_B$ are the labels of nodes A and B, respectively, and $l_e$ represents the edge label.

\noindent
\textbf{Storage Blocks Division.}
The design of sketches focuses on query efficiency and low storage space costs.
Our key idea is to cluster nodes with the same vertex labels together, so that the vertex labels can be encoded without taking up additional space.
Therefore, we divide the storage matrix into \zyl{$n \times n$ blocks ($n=d/b$)}, \cy{where $d$ represents the width of the whole matrix as in GSS, }and each block is a submatrix of width $b$. 
For an incoming item $e$, we first use a hash function $H(\cdot)$ to map the two vertex labels $l_A$ and $l_B$ into a value range \zzzyl{$[0,n)$} to locate a storage block. 
Then, we use $H(\cdot)$ again to map $A$ and $B$ into values in $[0,b)$ to locate a cell $E$ inside the selected block for storing $e$.

Through the above two-level hashing technique, we can not only complete the encoding of vertex labels without increasing the storage cost, but also ensure that the sketch can support subsequent aggregation queries based on vertex labels very efficiently.
Figure \ref{fig:blocks} illustrates the idea of Storage Blocks.

\begin{figure}[htbp]
\centering
\setlength{\abovecaptionskip}{0.1cm} 
\setlength{\belowcaptionskip}{-0.2cm}   
\includegraphics[width=2.5in]{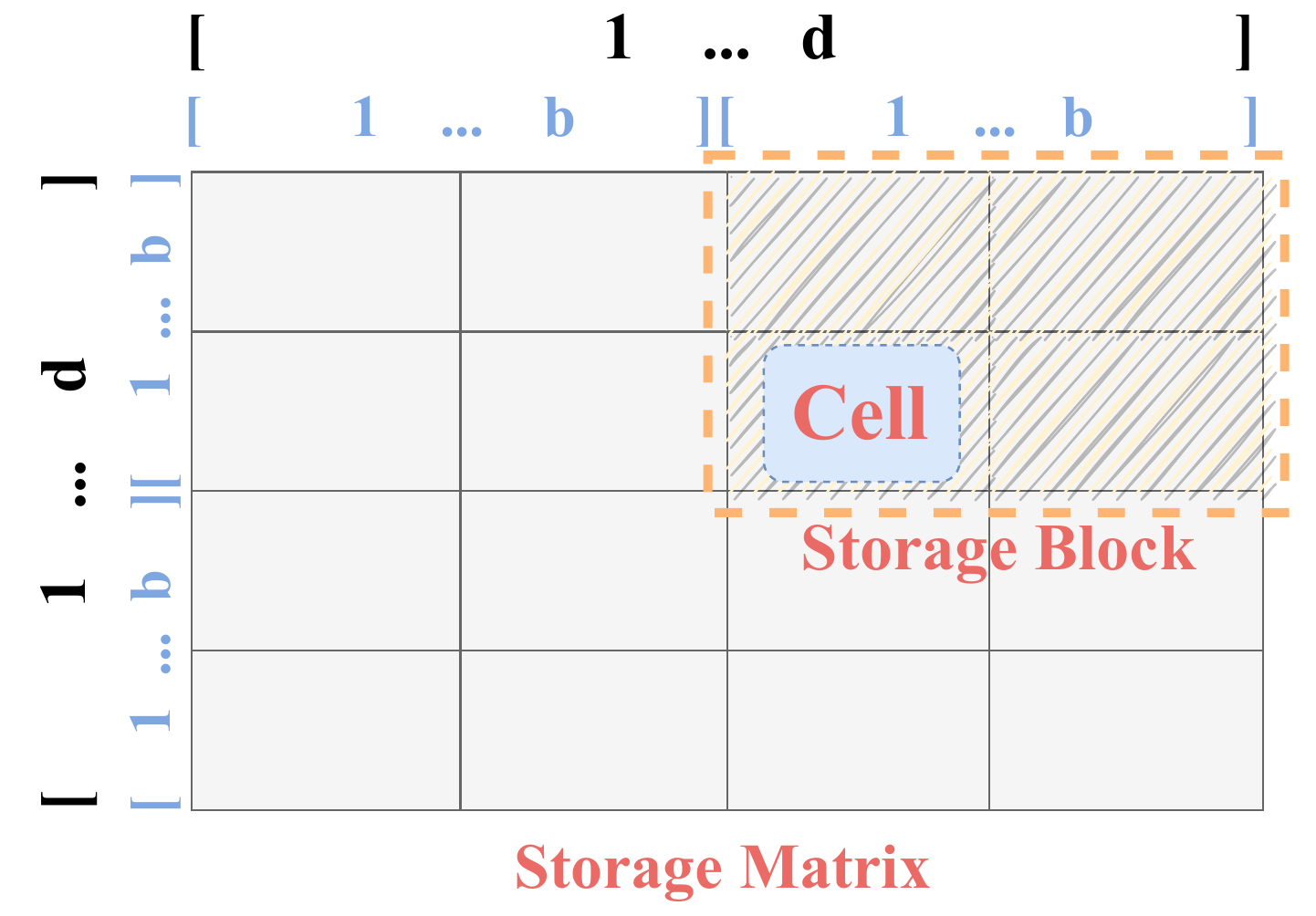} 
\caption{The illustration of Storage Blocks} 
\label{fig:blocks}
\end{figure}

\noindent
\textbf{Dual Counters.} 
As described in Section \ref{sec:graphstreams}, an edge may appear multiple times between the same pair of nodes, possibly with different weights or labels at different times.
The total weight of an edge in graph $G$ is the sum of all item weights sharing the same pair of endpoints. In order to effectively retain heterogeneous information, our sketch must be able to differentiate edge labels and store the respective weights of different edge labels.

We propose the Dual Counters mechanism to accurately record the edge weights of different types by maintaining two different counters in a matrix cell $E$. The content recorded by the first counter $C$ is the number of edges hitting the target cell. This counter is used for answering queries without label restrictions effectively. The second counter $P$ records the product of prime numbers to efficiently encode different types and different numbers of edges in one number. To achieve this, we give a list of prime numbers of length $c$ and map each distinct edge label to a unique prime number. For an incoming item $e$ with edge label $l_e$, we use $H(\cdot)$ to map $l_e$ into a value range $[0,c)$ to get a prime number representation for this particular edge label. Then the selected prime number is multiplied into counter $P$, so that the weights are superimposed.
Due to the uniqueness of factorization of the product of prime numbers, we can easily get the weights corresponding to each edge label through the counter $P$. Thus, the sketch can handle multi-edges in a very space-efficient manner.
\zyl{Note that when the sliding window is not introduced, we will process the counter $P$ into a list of great numbers by setting a threshold. In the application of practical scenarios, the storage of the counter $P$ is usually done by a great number under a reasonable window size.}
The Dual Counters mechanism 
applies to both unweighted and weighted graph streams (weights must be integers). For weighted graph streams, we only need to repeat the above operations according to the weights, which is illustrated later in Example 2.

\subsection{Timestamp Incorporated \yl{Matrix}}
In real applications, graph streams are updated at a high speed. With the continuous evolution of graph streams, the existence of old edges can have a detrimental effect on data analysis of the current moment. 
Sliding windows are a common technique~\cite{crouch2013dynamic, datar2002maintaining, kumar2015maintaining}.
Therefore, we devise a sliding window scheme to automatically handle edge deletions, ensuring the timeliness of graph streams to support the subsequent time-sensitive queries.

\noindent
\textbf{Sliding Windows.} 
Assuming that the size of a sliding window is $W$ time units, our sketch only maintains items that arrive after $t-W$, where $t$ is the current time, and items too old will be automatically removed.
Real-world queries are mostly related to time periods rather than specific time points.
Therefore, we subdivide the sliding window into $k$ subwindows according to the granularity demand of the analysis of the graph stream. Knowing $W$ and $k$, we can easily compute the size of each subwindow as $\frac{W}{k}$. 
\yl{\cy{In order to store required information in a storage-efficient manner, there is no need to store every item's timestamp. }
Instead, we propose to store a "lastT" timestamp, which represents the start time of the most recent subwindow. Although only one timestamp is recorded in the entire sketch, we have the ability to reason about the overall time region based on the size of the subwindows and \cy{the whole} window.}
Let the start time of the most recent subwindow be $t_n$. \cy{Whenever }
the current time $t \geq t_n+\frac{W}{k}$, we start a new subwindow with time $t$ and remove the oldest one. In this way, we can succinctly and effectively support sliding windows \cy{with a predefined granularity}.
\yl{Figure \ref{fig:windows2} illustrates the sliding process.}

\begin{figure}[htbp]
\centering
\setlength{\abovecaptionskip}{0.1cm} 
\setlength{\belowcaptionskip}{-0.2cm}   
\includegraphics[width=3in]{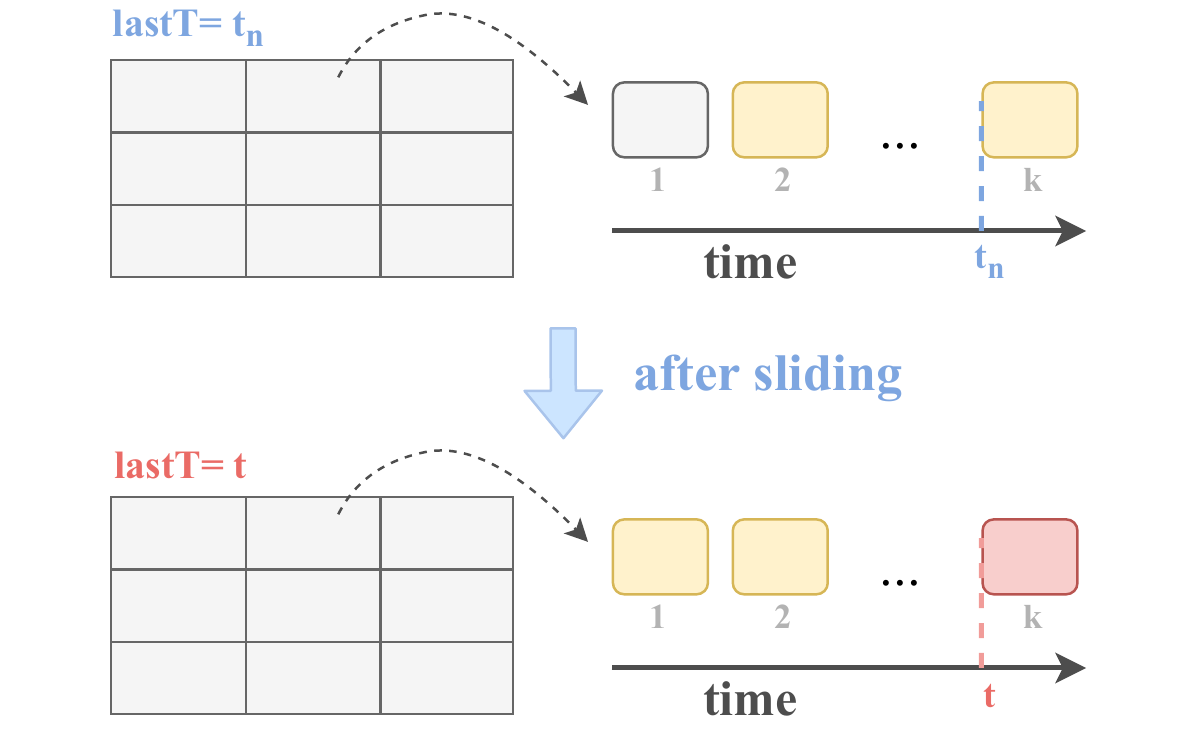} 
\caption{The sliding process of the storage matrix}
\label{fig:windows2}
\end{figure}


\begin{figure}[htbp]
\centering
\setlength{\abovecaptionskip}{0.1cm} 
\setlength{\belowcaptionskip}{-0.2cm}   
\includegraphics[width=3.2in]{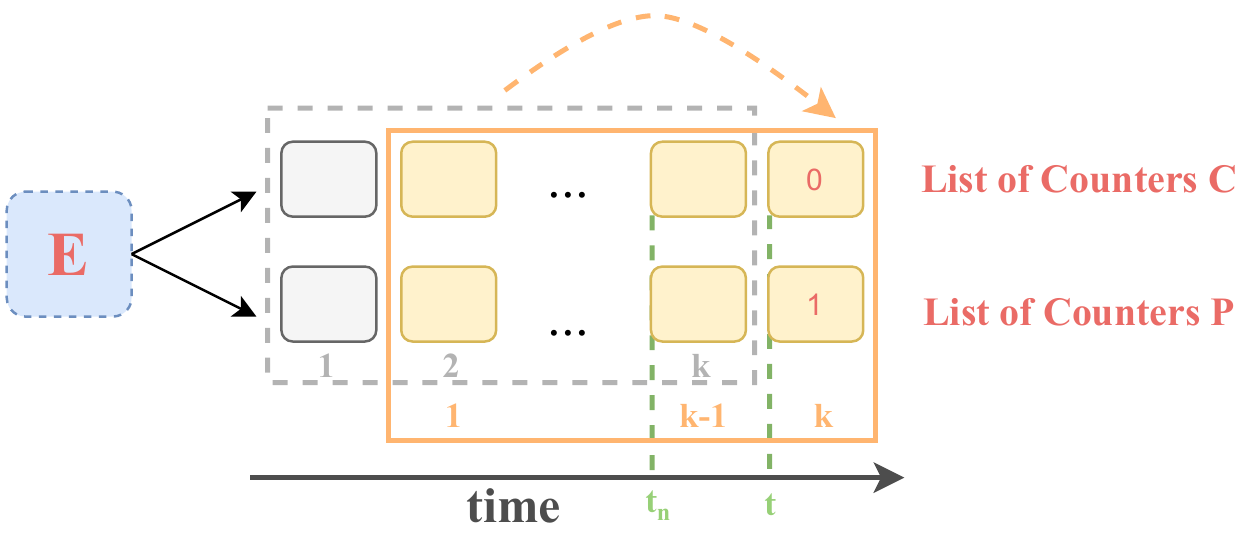} 
\caption{The data structure of a matrix cell after introducing sliding windows}
\label{fig:windows}
\end{figure}

From Figure \ref{fig:windows}, we can see that the Dual Counters mechanism can be easily integrated with sliding windows. For each subwindow, we assign two counters to record the total weight of the cell and the respective weights of different edge labels. Consequently, two counter lists of length $k$ are maintained in each matrix cell.

\subsection{\yl{The Additional Pool}} \label{sec:pool}
In addition to the storage matrix, we need a structure, named the \cy{a}dditional pool, to accommodate the conflicting items in the matrix. \cy{Similar as the adjacency list buffer used in GSS, }
the \cy{a}dditional \cy{p}ool should be a ``catch-all'' structure.
Specifically, we redesign the structure of the \cy{a}dditional pool \cy{as illustrated }
in Figure \ref{fig:pool}, in which the hash value $H(v)$ of a vertex $v$ serves as its identifier. 
We use an adjacency list to store the pointing relationship between vertices, while the weights between them are maintained using an array and a double-ended queue by storing the \cy{indexes (idx)}. This allows for fast sliding of subwindows and also ensures that the pool can slide synchronously with the storage matrix.


\begin{figure}[htbp]
\centering
\includegraphics[width=0.47\textwidth]{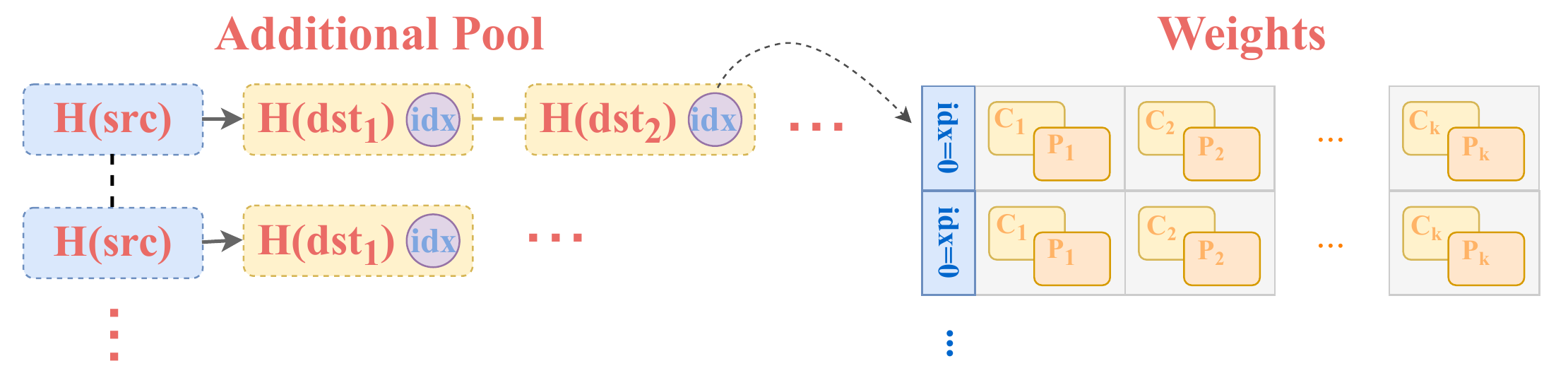} 
\caption{The structure of Additional Pool} 
\label{fig:pool}
\end{figure}

\subsection{LSkecth: the Full Framework.}
With various key components of LSketch in place, we are ready to see how they holistically work together.
\yl{Inspired by GSS, we adapt the Twin Cells and Square Hashing techniques into LSketch to further increase the capacity of the storage matrix to cope with uneven space consumption due to node degree skewness.
}
If the success rate of item insertions in the matrix increases, the overall update and query efficiency of LSketch will be significantly improved, as going through the \cy{a}dditional \cy{p}ool would be less efficient due to the list search cost.
\yl{Introducing the twin cells strategy is relatively simple\cy{. We} just \cy{need to} make sure that the fingerprint pair and the corresponding counter lists in a cell are divided into two twin segments, as illustrated in Figure \ref{fig:personality}. Next, we will talk about how to apply the square hashing strategy.}



\begin{figure}[htbp]
\centering
\setlength{\abovecaptionskip}{0.1cm} 
\setlength{\belowcaptionskip}{-0.2cm}   
\includegraphics[width=3.2in]{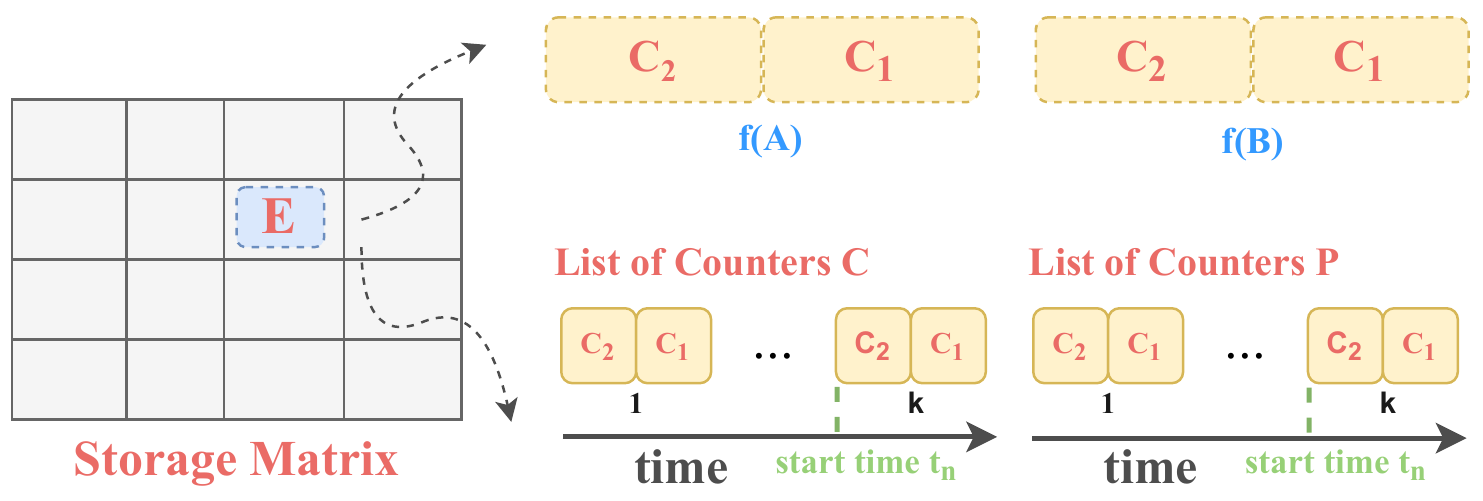} 
\caption{Illustrating twin cells in the matrix}
\label{fig:personality}
\end{figure}


\yl{When an item \cy{arrives}
, we first} use the same method as described in Section \ref{sec:block} to locate a certain storage block through the vertex labels of the two endpoints. All subsequent analyses are limited within this storage block.

In a nutshell, for a vertex $v$, we get its hash representation $H(v) = (s(v), f(v))$ using a hash function $H(\cdot)$. Then, \yl{referring to the idea of GSS, we use linear congruence method \cite{l1999tables} }to generate an address candidate list $\{s_1(v), s_2(v), ..., s_r(v)\}$ for $v$, where $0 \leq s_i(v)< b$.
Thus, edge $(A, B)$ has $r \times r$ candidate cells to be inserted into in the selected storage block according to the candidate lists of its two endpoints, which can be expressed as $\{(s_i(A), s_j(B))| 1\leq i\leq r, 1\leq j\leq r\}$, where $s_i(A)$ indicates a row coordinate and $s_i(B)$ indicates a column coordinate.
Once the attempt is successful, the edge will be stored in the first valid cell;
otherwise it needs to be placed in the \cy{a}dditional \cy{p}ool.

Obviously, to perform the item identity matching and subsequent queries successfully, the matrix cell needs to record the selection of the address candidate list, and we define it as an index pair $(i_r, i_c)$.
This means that for the from/to vertex $A/B$, the $i_r$-th/$i_c$-th address is selected, thus the item is stored in the $(s_{i_r}(A), s_{i_c}(B))$ position of the storage block.

\cy{We need to }
note that with enabling of the twin cells strategy, the inspection of each cell involves the matching attempts of two sets of index pairs and fingerprint pairs.
We use Figure \ref{fig:big} to illustrate the model with candidate Lists. For the sake of clarity, we omit the twin cells strategy and the sliding windows model here. 
The selected storage block is shown in the right plot of Figure \ref{fig:big}, in which $r$ is set to 2, resulting in a total of 4 candidate cells.

\begin{figure}[htbp]
\centering
\setlength{\abovecaptionskip}{0.1cm} 
\setlength{\belowcaptionskip}{-0.2cm}   
\includegraphics[width=3.4in]{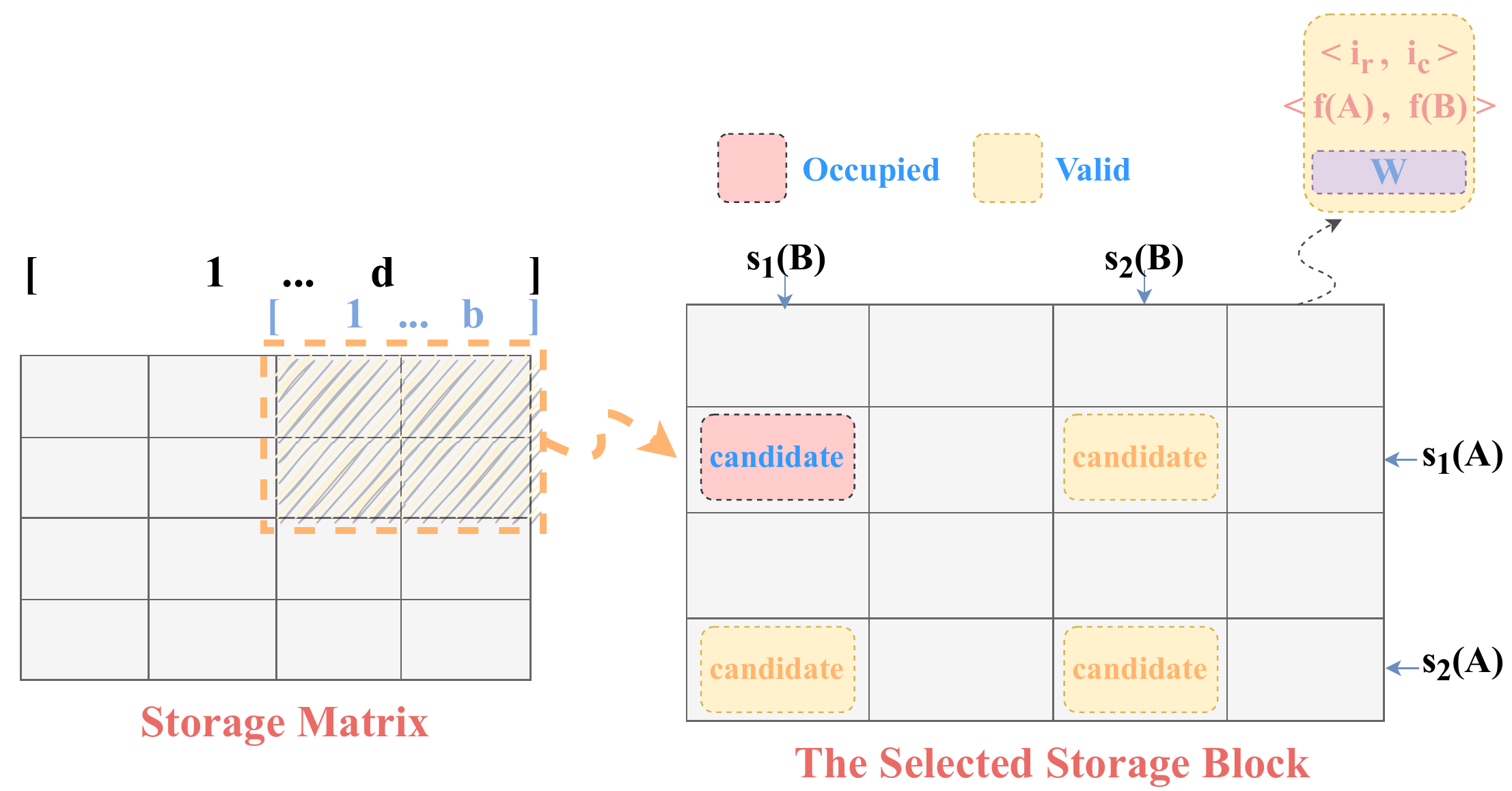} 
\caption{The illustration of the Candidate Lists} 
\label{fig:big}
\end{figure}



\begin{algorithm}[htbp]
	\caption{Precompute}
	\label{al:precompute}
	\LinesNumbered
	\small{ \KwIn{a vertex $A$ and its corresponding label $l_A$}  }
	\KwOut{the hash value $H(A)$, the initial address $s(A)$, the fingerprint $f(A)$, and the candidate list $l_i(A) (1\leq i \leq r)$ }
	Initialize the width of the storage matrix $d$, the width of the storage block $b$, the fingerprint length $F$, the candidate list length $r$, and hash function $H(\cdot)$ \\
	
    $\zzzyl{m_A} = H(l_A) \% \frac{d}{b}$\\
    $s(A)= \lfloor \frac{H(A)}{F} \rfloor$\\
    $f(A)= H(A)\% F$ \\
    $l_1(A)=(T\times f(A)+I)\%M$ \\
    \For{$i \leftarrow 2 \ldots r$}
    {
        $l_i(A)=(T\times l_{i-1}(A)+I)\%M $  \\
    }
	return $H(A)$, $s(A)$, $f(A)$, and $l_i(A) (1\leq i \leq r)$
\end{algorithm}

After getting familiar with all constituent parts, we are ready to show the whole algorithm.
Algorithm \ref{al:precompute} shows the process of calculating the address candidate list using vertex identifiers and labels. 
We first calculate the \zyl{index} of the storage block in a certain dimension according to the vertex label (line 2), and then calculate the initial address and the fingerprint of the vertex (lines 3-4). The address candidate list of the given vertex is generated using Equations \ref{linear} and \ref{linear2} (lines 5-7).

Algorithm \ref{al:lsketch} illustrates the complete process of how to insert items into LSketch. We mainly show the process of items inserting into the storage matrix, while the \cy{a}dditional \cy{p}ool \cy{related} insertion is similar but simpler. 
We first initialize the LSketch (lines 1-4). Then we handle the item insertion upon each item's arriving.
At the beginning, according to the timestamp carried by the incoming item, we first check if we need to start a new subwindow and perform the window sliding (lines 6-9).
Then, we use other information to find possible insertion positions of the item.
Calling Algorithm \ref{al:precompute} on vertices $A$ and $B$, we can get their fingerprints and candidate lists (line 10).

After that, we generate the sampling sequences according to the strategy described earlier in Section \yl{\ref{sec:gss}} (lines 11-15). 
\zyl{Finally, we check the $s$ cells in turn. If an empty cell is encountered, insert it directly; otherwise, we will update the values if the index pair and the fingerprint pair of the cell completely match the inserted ones (lines 16-23).
Once the insertion or update is complete, this round is finished.}
Otherwise, we place the item into the \cy{a}dditional \cy{p}ool if all the above attempts fail (lines 24-25).

\begin{algorithm}[htbp]
	\caption{\zzzyl{LSketch}: Labeled Graph Stream Sketch}
	\label{al:lsketch}
	\LinesNumbered
	\small{ \KwIn{a graph stream $G$}  }
	\KwOut{\zyl{LSketch $S_k$, with the storage matrix $S$} and an addtional pool $AP$ }
	Initialize the width of the storage matrix $d$, the width of the storage block $b$, the fingerprint length $F$, the candidate list length $r$ and its sampling length $s$, along with window size $W$ and subwindow size $W_s$. Also set the start time of the latest subwindow $t_n$ to current time $t$.\\
	$S\leftarrow$ an empty sketch \\
	$\zyl{P_r} \leftarrow$  a predefined list of prime numbers of length $c$\\
	$k \leftarrow W/W_s$, the number of subwindows \\
	
	\For{each arriving item $e=(A,B;l_A,l_B,l_e;w;t)$}
	{	
	    \uIf{$t_n+W_s \leq t$}
		{
			\For{$i \leftarrow 2 \ldots k$}
            {
                \zyl{$S[i-1] \leftarrow S[i]$}  \\
            }
            $t_n = t$
		}
        
        $Precompute(A,l_A)$, $Precompute(B,l_B)$  \\
       
        $Sp_0(e) = f(A)+f(B)$ \\
        \For{$i \leftarrow 1 \ldots s$}
        {
            $Sp_i(e) = (T\times Sp_{i-1}(e) + I)\%M$ \\
            $A_i =  \lfloor \frac{Sp_i(e)}{r} \rfloor \% r$, $B_i =  Sp_i(e)\%r $   \\
            $p_1 = (s(A)+l_{A_i}(A)) \% b$, $p_2 = (s(B)+l_{B_i}(B)) \% b$  \\
    		$E = S[\zyl{m_A}\times b+p_1][\zyl{m_B}\times b+p_2]$  \\
            Traverse the twin cells in turn \\
                \uIf{the stored index pair and fingerprint pair match the ones to be inserted \zyl{or the cell is empty}}
                {
                    $P_e = \zyl{P_r}[H(l_e)\%c]$  \\
                    \For{$i \leftarrow 1 \ldots w$}
                    {
                        $E[k].C = E[k].C +1$  \\
                        $E[k].P = E[k].P \times P_e$ \\
                    }
                    $S[\zyl{m_A}\times b+p_1][\zyl{m_B}\times b+p_2] = E$ \\
                }
        }
        \uIf{all insertion attempts in the storage matrix fail}
        {Insert $e$ into the additional pool $AP$}
	}
	return $S$
\end{algorithm}

Next, we show some examples of LSketch to better illustrate how the algorithms work.

\begin{example}
Suppose a graph stream is composed of the following items: $(a,b;l_2,l_1,l_{e1};3;t1)$, $(a,c;l_2,l_1,l_{e1};1;t2)$, $(b,d;l_1,l_2,l_{e2};2;t3)$, $(b,e;l_1,l_1,l_{e1};1;t4)$, \\ $(c,b;l_1,l_1,l_{e2};2;t5)$, and $(e,c;l_1,l_1,l_{e1};1;t6)$. We set $d=6, b=3, F=8, r=2$, and $s=2$ at this point. 

First, we compute the corresponding hash values and calculate the final candidate addresses of each vertex, which is shown in the left part of Figure \ref{fig:example1} (we omit the process of calculating address candidate list here for clarity).
Based on the hash values of the corresponding vertex labels, we can know which block each item should be stored in. For example, items $(b,e)$, $(c,b)$ and $(e,c)$ should be placed in the upper left storage block.
Next, we update the sketch with the insertion of each item.
For item $(a,b;l_2,l_1,l_{e1};3;t1)$, it is placed in the $0$-th/$1$-st cell of the lower left storage block with index pair $(1,1)$ and fingerprint pair (1,4). For item $(a,c;l_2,l_1,l_{e1};1;t2)$, we first check the $0$-th/$1$-st cell of the lower left storage block according to the two endpoints' candidate address list. Note that we omit the twin cells strategy here to better illustrate how to deal with cell conflicts and to search through the address list. Since the cell is occupied by other items (as the fingerprints do not match), we check the $0$-th/$2$-nd cell of the lower left storage block and insert it with the index pair $(1,2)$ successfully (note that in practice, it could certainly succeed during the first insertion attempt when there is no conflict).
The insertion of other items is similar, and the final result is shown in the right plot of Figure \ref{fig:example1}.
Since this example focuses on showing how to locate items, we omit the operations of how weights change in it.
\begin{figure}[htbp]
\centering
\setlength{\abovecaptionskip}{0.1cm} 
\setlength{\belowcaptionskip}{-0.2cm}   
\includegraphics[width=3.5in]{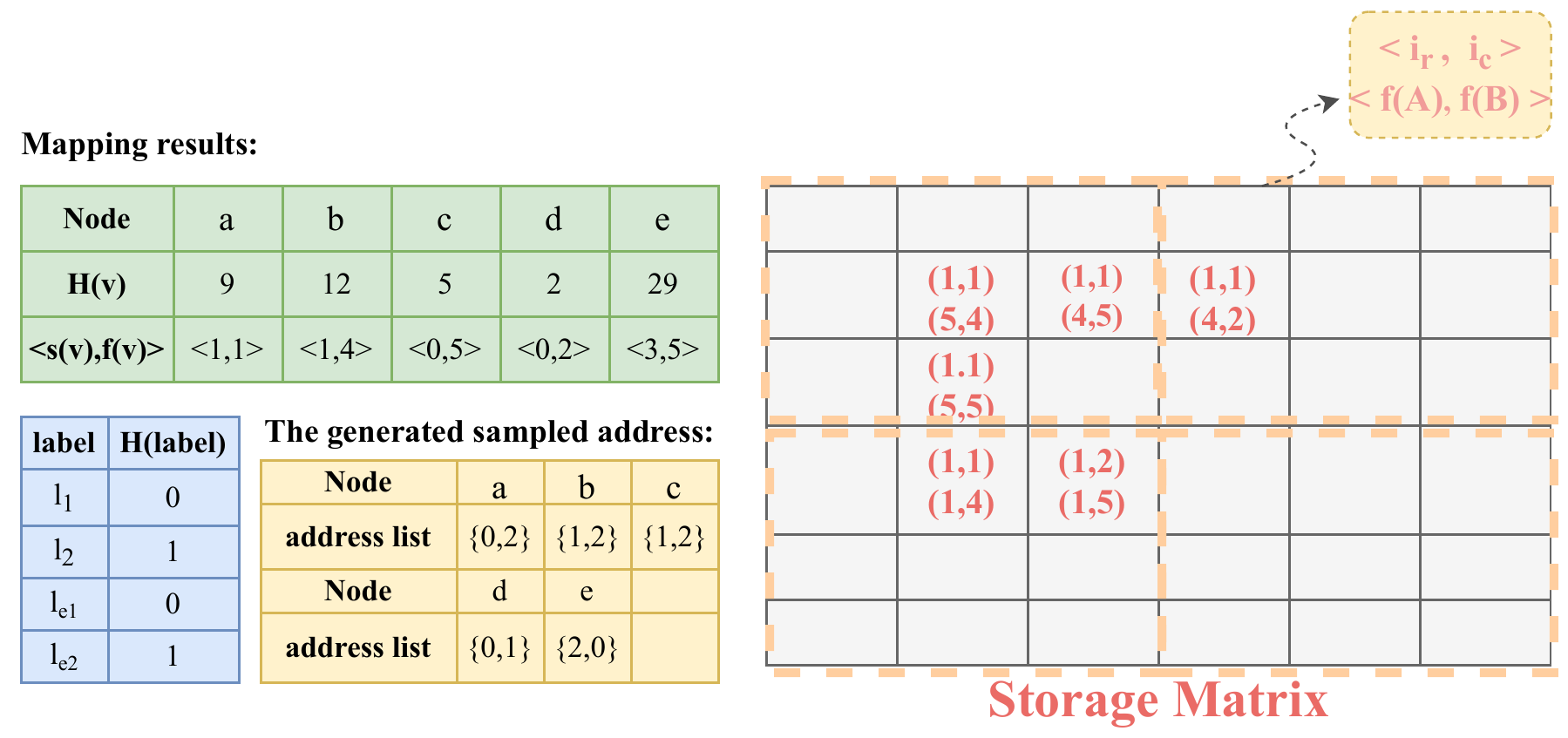} 
\caption{The first example of LSketch} 
\label{fig:example1}
\end{figure}
\end{example}

\begin{example}
\begin{figure}[htbp]
\centering
\setlength{\abovecaptionskip}{0.1cm} 
\setlength{\belowcaptionskip}{-0.2cm}   
\includegraphics[width=3.5in]{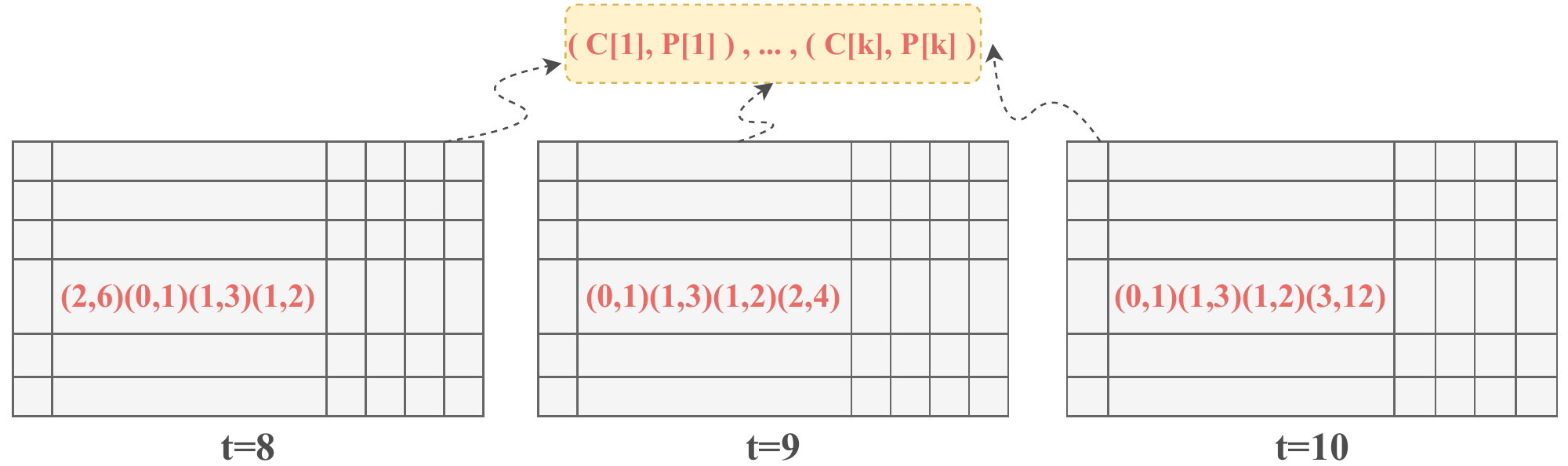} 
\caption{The second example of LSketch} 
\label{fig:example2}
\end{figure}

Next, we focus on showing how to handle the weights variation with the sliding window model in LSketch.

We set $W=8$ and $W_s=2$, and thus each cell maintains two counter lists of length 4. Also, we give the predefined prime number list as \zzzyl{$P_r=[2,3]$} with $c=2$.
Suppose we have located the item with the vertex pair $(a,b)$ to the $0$-th/$1$-st position of the lower left storage block (the same setting as in the previous example), and the current state of time $8$ is shown in Figure \ref{fig:example2}.

Now we continue to process items of vertex pair $(a,b)$. At time $9$, suppose item $(a,b;l_2,l_1,l_{e1};2;9)$ arrives. Since the first subwindow expires, we slide it and start a new one.
Given the hash results, where edge label $l_{e1}$ corresponds to the prime number \zzzyl{$P_r[0]=2$}, we update the new subwindow from $(0,1)$ to $(0+1+1,1*2*2)$. At time 10, suppose item $(a,b;l_2,l_1,l_{e2};1;10)$ arrives, so the latest subwindow is updated to $(2+1,4*3)$, given the prime number representation \zzzyl{$P_r[1]=3$}.
\end{example}

\subsection{Further Improvements}
So far, LSketch is able to cope well with the heterogeneous graph streams that are updated at high speed in real-world scenarios. However, it is observed that there are some \cy{circumstances }
where the labels of the vertices are not evenly distributed. That is, most vertices in the graph stream are of the same label. If we continue to apply the previous idea of Storage Blocks Division, it will lead to a situation where a certain block is severely \cy{congested }
while other blocks are stored sparsely.
Therefore, we propose another strategy, called "Skewed Blocking", for such extreme cases to improve the defects of the previous uniform blocking.

\noindent
\textbf{Skewed Blocking.} 
In order to efficiently encode vertex labels into sketches and support subsequent queries, we still use the idea of blocking. However, the size of each matrix blocks should be set upon demand. In short, instead of dividing the matrix evenly, we set the proportional distribution of labels according to a predefined scenario.

Given a dataset containing two vertex labels, Figure \ref{fig:even2skewed} gives a sample illustration of candidate lists when inserting an edge under both uniform and skewed blocking strategies, where the ratio of the two vertex labels is set to 3:7 under the skewed blocking strategy. It can be seen that when faced with extreme \cy{unbalanced} datasets, the strategy of skewed blocking can well improve the success rate of edge insertion into the matrix. \cy{In order to get the distribution characteristics of the graph stream, we can collect the data for a short period of time and choose from uniform blocking or skewed blocking according to the distribution characteristics of the data.}

\begin{figure}[htbp]
\centering
\setlength{\abovecaptionskip}{0.1cm} 
\setlength{\belowcaptionskip}{-0.5cm}   
\subfigure[The uniform blocking]{
\begin{minipage}[htbp]{0.47\textwidth}
\includegraphics[width=\textwidth]{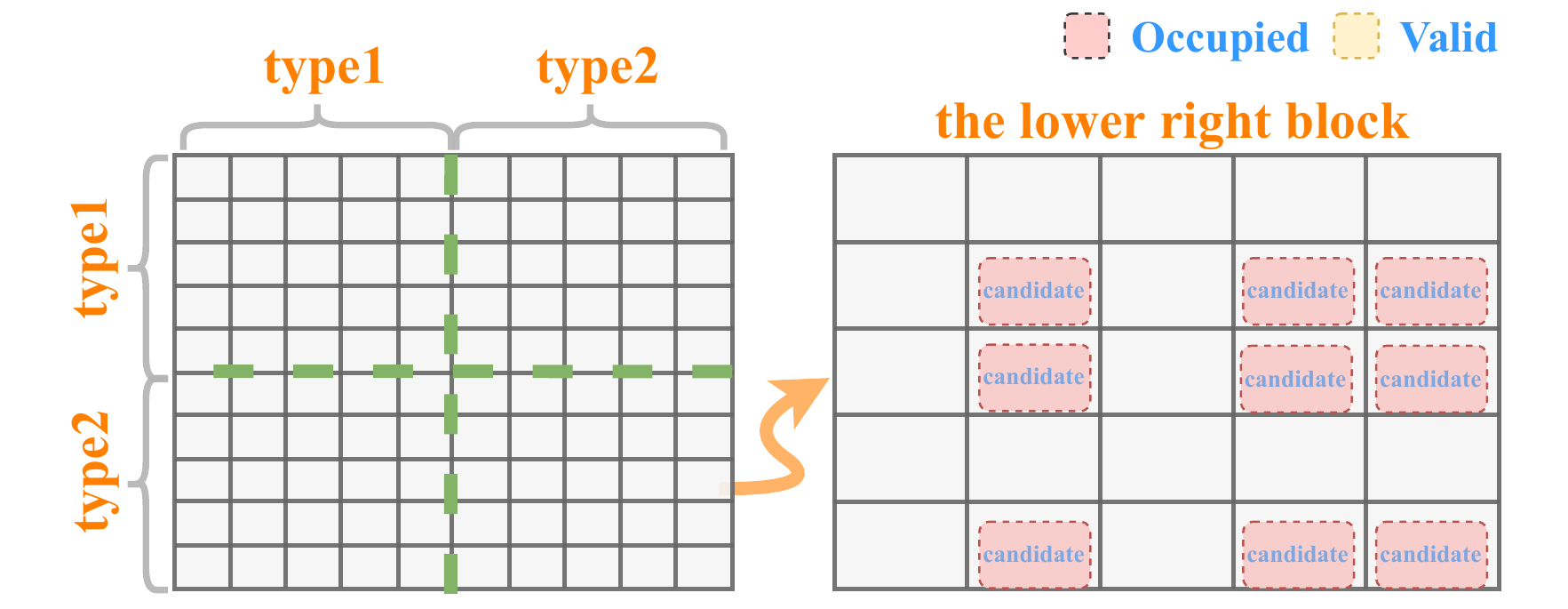}
\end{minipage}
}
\\
\subfigure[The skewed blocking]{
\begin{minipage}[htbp]{0.47\textwidth}
\includegraphics[width=\textwidth]{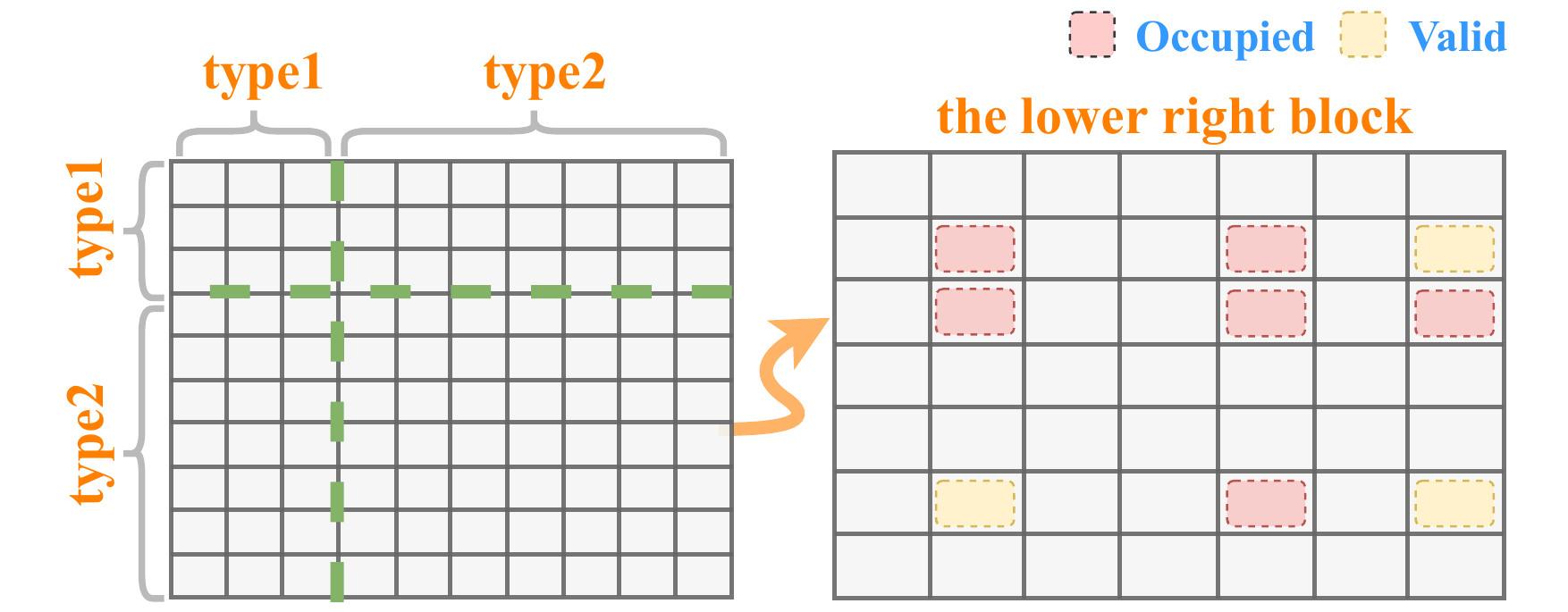}
\end{minipage}
}
\caption{Illustrating the comparison of uniform and skewed blocking}
\label{fig:even2skewed}
\end{figure}

\subsection{Storage and time complexity analysis} 
In practice, we implement the storage matrix by pointers, so only existing edges occupy the storage space. Along with the additional pool, LSketch stores multiple edges (where edges have the same endpoints but may have different edge types, and may arrive at different timestamps) in one storage cell, so the storage space needed is at most $k\times |E_s|$, where $|E_s|$ is the number of distinct edges (edges with the same endpoints are counted only once) within the sliding window, and $k$ is the number of subwindows. This is expected to be much less than the actual edge number $|E|$. Thus LSketch is very efficient in the aspect of storage, especially for graph stream with a high incoming rate and many duplicate edges. 
As for time consumption, 
LSketch takes $O(1)$ time to locate the storage block, and we check at most $s$ sample cells when inserting edges into the storage block, which consumes $O(s)$ time. Although the additional pool takes linear time to insert, the chance of actually \cy{needed }
to insert into it is almost negligible, with our various techniques to avoid doing so.
Since $s$ is usually a small constant, the insertion time complexity is $O(1)$.


\section{LSketch Powered Graph Queries}
\begin{algorithm}[htbp]
	\caption{GetWeightsInM}
	\label{al:getw}
	\LinesNumbered
	\small{ \KwIn{the selected matrix cell $E$ and its child option $shifter$, an edge label $l_e$ (optional)}  }
	\KwOut{$w$: the total weight of edges in one child segment of $E$, $w_l$: the weight of edges with label $l_e$ in one child segment of $E$}
    $P_e = \zyl{P_r}[H(l_e)\%c] $ //get the prime number representation of $l_e$ if provided\\ 
    $w=0, w_l=0$  \\
    \For{$i \leftarrow 1 \ldots k$}
    {
        $w=w+E[shifter][i].C$ \\
        $temp_p=E[shifter][i].P$  \\
        \While{$temp_p\%P_e =0$} 
	    {
	        $w_l=w_l+1$ \\
	        $temp_p=temp_p/P_e$ \\
	    }
    }
    return $w, w_l$
\end{algorithm}
Since LSketch keeps the connections between nodes, it is able to take care of all graph structure based queries. 
In this section, we will show some of them to illustrate how LSketch \yl{using the uniform blocking strategy} supports various queries, including structure based queries and time-sensitive queries.
The query methods on LSketch using the skewed blocking strategy can be easily extended, so we \cy{omit them for space saving.}
First, we introduce a basic algorithm \textsc{GetWeightsInM}, which is the basis for subsequent queries.

The algorithm \textsc{GetWeightsInM} is used to compute the total weight and the weight with a specific edge label, given a particular child segment in a designated storage cell.
In practice, we use bit operations to obtain the weights of a certain child segment in a cell. But for convenience, we treat it as a list in the pseudocode.
Therefore, by traversing $k$ subwindows and processing the counters in turn (lines 3-8), we can easily obtain the required weights.


\zyl{To facilitate the accuracy analysis of the following queries in this section, we first discuss the probability that an edge $e$ suffers from
edge collision.}

\begin{theorem} \label{edgecollision}
The probability that an edge encounters an edge collision $\tilde{P}$ is $1-e^{-\frac{(q_1 + L q_2 + L^{2} q_3)(|E|-d_v) + (DLh_1+DL^{2}h_2)d_v}{D^{2}L^{2}}}$, where $|E|$ is the number of edges and $d_v$ is the out-degree of vertex $v$ \cs{(the meanings of other parameters will be demonstrated in the proof, and followed by an explanatory example.)}.
\end{theorem}

\begin{proof}
According to the construction method of LSketch, we know that edge collisions may only occur in the initial hashing stage, since all subsequent index pairs and fingerprint pairs ensure that edges with different hash values cannot occupy the same storage space.  
As stated before, the hash range is $[0,D)$ for vertices and \zyl{$[0,L)$} for vertex labels.

Locating an edge includes the hashing of vertex identifiers and vertex labels.
For an edge that has no common endpoints with a given edge $e$, it will only collide with $e$ when the starting vertex, ending vertex and the two vertex labels all collide with the corresponding ones of $e$.
There are three cases for the collision of vertex labels.
a) The corresponding starting vertex labels and the ending vertex labels are both different while their hash values are the same.
b) The vertex labels corresponding to the starting vertices or the ending vertices are different while their hash values are the same.
c) The corresponding vertex labels of the starting and the ending vertices are the same.
Thus the probability of collision can be expressed as:
\zyl{
\begin{equation} 
   p_1=\frac{1}{D^{2}}\times (\frac{1}{L^{2}}\times q_1 + \frac{1}{L}\times q_2 + q_3)
\end{equation}
, where $q_1$, $q_2$, and $q_3$ are the probability of the above three cases respectively and $q_1 + q_2 + q_3=1$.
And the probability that none of the other $|E|-d_v$ edges collides with $e$ is:
\begin{equation} 
   P_1=(1-p_1)^{|E|-d_v}
\end{equation}
}

As for the $d_v$ edges sharing one common endpoint with $e$, there are two cases for the vertex label to collide with another vertex.
a) The vertex labels of the non-shared vertices are different while their hash values are the same.
b) The vertex labels of the non-shared vertices are indeed the same.
Therefore, the probability that those edges collide with $e$ is:
\zyl{
\begin{equation} 
   p_2=\frac{1}{D}\times (\frac{1}{L} \times h_1 + h_2)
\end{equation}
, where $h_1$ and $h_2$ are the probability of the above two cases respectively and $h_1+h_2=1$.

And the probability that none of the $d_v$ edges collides with $e$ is:
\begin{equation} 
   P_2=(1-p_2)^{d_v}
\end{equation}
}

In summary, the probability that none of the $|E|$ edges collides with $e$ is:
\zyl{
\begin{align} 
    P= & P_1\times P_2 
     =(1-p_1)^{|E|-d_v}\times (1-p_2)^{d_v} \nonumber \\ &
     = e^{-p_1\times (|E|-d_v)}\times e^{-p_2\times d_v} \nonumber \\ &
     = e^{-\frac{q_1 + L q_2 + L^{2} q_3}{D^{2}L^{2}} \times (|E|-d_v)} \times e^{-\frac{h_1 + L h_2}{DL}\times d_v} \nonumber \\ &
     = e^{-\frac{(q_1 + L q_2 + L^{2} q_3)(|E|-d_v) + (DLh_1+DL^{2}h_2)d_v}{D^{2}L^{2}}}
\end{align}
}
and the probability $\tilde{P}$ that an edge encounters an edge collision is $1-P$.

\end{proof}
\zyl{For the sake of analysis, we assume that the node labels obey a uniform distribution, which leads to 
\begin{equation} 
   \left\{
             \begin{array}{lr}
             q_1 = \frac{(l-1)^{2}}{l^{2}}  \\
             q_2 = \frac{2(l-1)}{l^{2}}  \\
             q_3 = \frac{1}{l^{2}}  \\
             h_1 = \frac{l-1}{l} \\
             h_2 = \frac{1}{l} \\
             \end{array}
    \right.
\end{equation}
, where $l$ is the number of node labels. Thus, 
\begin{align} 
    P=  e^{-(\frac{L+l-1}{DLl})^{2}(|E|-d_v) - \frac{L+l-1}{DLl}d_v}
\end{align}
}

\zyl{In LSketch, \cs{supppose }we have $D = d \times F$, \cs{and }$L = t \times F$\cs{,} where $d$ is the width of the matrix and $F$ is the range of fingerprints.
Suppose we set $F = 256$ (with 8-bit fingerprints), \cs{and }$d=1000$. When working with a dataset with $|E| = 5 \times 10^{5}$ and $t=2$, we can get the probability $P=0.9996$ with $d_v = 200$. \cs{In another word, the probability that an edge encounters an edge collision is only 0.0004, which well meets the accuracy requirement in real applications.}

}

\subsection{Vertex Queries}
Vertex queries include outgoing edge weight queries and incoming edge weight queries of a vertex $v$, which are to estimate its outgoing weight or incoming weight (this is equivalent to out-degree and in-degree queries when we set all weights to 1). Additionally, the vertex queries can be restricted by an edge label $l_e$, thus the query result indicates the weight with a certain edge label sent/received by vertex $v$.
Here we only discuss the outgoing edge weight queries of a vertex $v$ since the incoming edge weight queries are similar.

\begin{algorithm}[htbp]
	\caption{AggregateVertexQueries}
	\label{al:node}
	\LinesNumbered
	\small{ \KwIn{a vertex $A$, a vertex label $l_A$, an edge label $l_e$ (optional)} }
	\KwOut{$w$: the outgoing weight of vertex \zzzyl{$A$}, $w_l$: the outgoing weight with label $l_e$ of vertex \zzzyl{$A$}; $sum$: the outgoing weight of all vertices with label $l_A$, $sum_l$: the outgoing weight with label $l_e$ of all vertices with label $l_A$}
	$w=0, w_l=0, sum=0, sum_l=0 $  \\
	
	$Precompute(A,l_A)$  \\
    \For{$i \leftarrow 1 \ldots r$}
    {
         $p_1 = (s(A)+l_{i}(A)) \% b$ \\
         \For{$j \leftarrow 1 \ldots d$}
         {
            \For{$shifter \leftarrow 1 \ldots 2$} 
            {
                \uIf{the stored index equals $i$ and the stored fingerprint matches $f(A)$ }
                {
                    {\tiny $w += GetWeightsInM(S[m\times b +p1][j], shifter) $}  \\
                    {\tiny $w_l += GetWeightsInM(S[m\times b +p1][j], shifter, l_e) $} \\
                }
            }
         }
    }  
    \For{\zyl{$i \leftarrow m\times b \ldots m\times b +b$}}
    {
        \For{$j \leftarrow 1 \ldots d$}
        {
            \For{$shifter \leftarrow 1 \ldots 2$}  
            {
                {\scriptsize $sum += GetWeightsInM(S[i][j], shifter) $ } \\
                {\scriptsize $sum_l += GetWeightsInM(S[i][j], shifter, l_e) $}  \\
            }
        }
    }
    $w= w + weight$ gained in additional pool \\
    $w_l=w_l + weight$ gained in additional pool \\
    $sum=sum + weight$ gained in additional pool \\
    $sum_l=sum_l + weight$ gained in additional pool \\
    return $w, w_l, sum, sum_l$
\end{algorithm}

The detailed algorithm is shown in Algorithm \ref{al:node}.
For the storage matrix, we first calculate the possible rows of vertex $v$ (lines 2-4) and then traverse all the columns. If the index and the fingerprint match successfully, we use \textsc{GetWeightsInM} to obtain the weights of all cells that meet the conditions (lines 5-9).
We also traverse the additional pool 
to find the leftovers. Recall the data structure described in Figure \ref{fig:pool}---the elements in the additional pool are linked by source nodes, and hence we need to examine the source nodes list.  
If the hash value of a linked-list head is the same as that of the queried vertex, then we have found a matched source node, and the total weight in that list should be added to $w$. The weight with label restriction is also obtained and added to $w_l$ (lines 15-16).

Furthermore, we can do the weight queries on a certain type of vertex since our sketch places vertices with the same vertex label (type) together.
We use the vertex label to get the starting row \zyl{ $m\times b$ (line 2) of the storage block, and thus the storage location of vertices with a certain label is $[m\times b,m\times b+b)$. }
For these $b$ rows, we traverse all the columns and use \textsc{GetWeightsInM} to sum the weights (lines 10-14).
Similarly, we traverse the additional pool to get the final result (lines 17-18).

Based on Theorem \ref{edgecollision}, for vertex queries, the accuracy of outgoing edge weight queries without label restriction is ensured by $P^{|V|-d_v}$, where $|V|$ denotes the number of vertices in $G$ and $d_v$ is the out-degree of the queried vertex.
The outgoing edge weight queries with label restriction involves the hashing of edge labels of range $[0,c)$, and thus its accuracy is ensured by $P^{|V|-d_v}\times (1-\frac{1}{c})^{l-1}$, where $l$ is the number of edge labels.

As for the query time, we need to traverse the possible $r$ rows of the queried vertex, sum up the weights from the $k$ subwindows of all matching cells, and finally supplement the weights in the additional pool. So the time complexity is $O(d\times r\times k + |AP|)$.

\subsection{Edge Queries}
Given an edge $e$ with endpoints $A$, $B$, and the corresponding vertex labels $l_A$ and $l_B$, an edge query estimates the edge weight of $e$. Particularly, the edge queries can be restricted by an edge label $l_e$, aiming to obtain the weight of edges of type $l_e$ between endpoints $A$ and $B$.
Similar to vertex queries, LSketch can handle the case where vertices are of certain types, e.g., edge queries between a vertex and vertices of a particular type, or between two particular types of vertices.

\begin{algorithm}[h]
	\caption{AggregateEdgeQueries}
	\label{al:edge}
	\LinesNumbered
	\small{ \KwIn{an edge $e=(A,B)$ with corresponding vertex labels $l_A$ and $l_B$, an edge label $l_e$ (optional)}  }
	\KwOut{$w_{vg}$: the weight between vertices $A$ and vertices of label $l_B$, $w_{vg\_l}$: the weight with label $l_e$ between vertices $A$ and vertices of label $l_B$ }
	
	$w_{vg}=0, w_{vg\_l}=0 $  \\
	$Precompute(A,l_A)$, $Precompute(B,l_B)$  \\
    \For{$i \leftarrow 1 \ldots r$}
    {
         $p_1 = (s(A)+l_{i}(A)) \% b$ \\
         \For{\zyl{$j \leftarrow m_B\times b \ldots m_B\times b+b$}}
         {
            \For{$shifter \leftarrow 1 \ldots 2$} 
            {
                $E = S[m\times b+p_1][j]$  \\
                \uIf{the stored index equals $i$ and the stored fingerprint matches $f(A)$ }
                {
                    {\scriptsize $w_{vg} += GetWeightsInM(E, shifter) $}  \\
                    {\scriptsize $w_{vg\_l} += GetWeightsInM(E, shifter, l_e) $} \\
                }
            }
         }
    }  
    $w_{vg}=w_{vg} + weight$ gained in additional pool \\
    $w_{vg\_l}=w_{vg\_l} + weight$ gained in additional pool \\
    return $w_{vg}, w_{vg\_l}$
\end{algorithm}

The detailed algorithm is shown in Algorithm \ref{al:edge}. 
The edge queries of a specific edge are similar to the insertion process; thus we only show the edge queries between a vertex and a particular type of vertices.
When performing edge queries between vertex $A$ and vertices with label $l_B$, we first traverse the $r$ rows of $A$'s possible locations generated by the candidate list technique. Then we calculate the starting column in the storage block of label $l_B$, and traverse columns from \zyl{$m_B\times b$ to $m_B\times b+b$}. If the index and the fingerprint are both matched, we sum the weights (lines 7-9). At last, we traverse the additional pool to get the final result (lines 10-11).


The accuracy of the edge queries without label restriction is ensured by exactly $P$ as calculated in Theorem \ref{edgecollision}, and the accuracy of the edge queries with edge label restriction is bounded by $P\times (1-\frac{1}{c})^{l-1}$, where $l$ is the number of edge labels and $c$ is the hash range of them.
Also, the time complexity of edge queries is the same as that of the insertion into LSketch, which is $O(1)$.

\subsection{Path Queries}
Given two vertices $A$ and $B$, a path query is to determine whether there is a reachable path from $A$ to $B$. Particularly, the path queries can be restricted by a specified edge label to find the existence of paths of certain types from $A$ to $B$.

Based on the structural information maintained by LSketch, our algorithm can be applied as a black box to any existing path reachability algorithm (or any other structure-based graph queries). We use BFS to illustrate the query as an example, and the detailed algorithm without edge label restriction is shown in Algorithm \ref{al:path}.

\begin{algorithm}[htbp]
	\caption{PathReachability}
	\label{al:path}
	\LinesNumbered
	\small{ \KwIn{vertices $A$ and $B$, and its corresponding vertex labels $l_A$ and $l_B$}  }
	\KwOut{whether there is a path from $A$ to $B$}
   
    $q = \{A\}$ // queue  \\ 
    $checked = \{A\}$   \\
    \While{$q$ is not empty}
    {
        $src$ = $q$.head \\
        
        $Precompute(src,l_{src})$, $Precompute(B,l_B)$  \\
        \For{$i \leftarrow 1 \ldots r$}   
        {
            $p_1 = (s(src)+l_{i}(src)) \% b$ \\
            \For{$j \leftarrow 1 \ldots r$}
            {
                $p_2 = (s(B)+l_{B_j}(B)) \% b$  \\
        		$E = S[\zyl{m_{src}}\times b+p_1][\zyl{m_B}\times b+p_2]$  \\
                \For{$shifter \leftarrow 1 \ldots 2$} 
                {
                    \uIf{the stored index pair and fingerprint pair match the ones of query vertices}
                    {
                        \uIf{ {\scriptsize$GetWeightsInM(E, shifter) >0$}}
                        {return true}
                    }
                }
            }
        }
        {\scriptsize Add all successors of $src$ in additional pool to $q$ and $checked$ } \\  
        \For{$i \leftarrow 1 \ldots r$}   
        {
            $p_1 = (s(src)+l_{i}(src)) \% b$ \\
            \For{$j \leftarrow 1 \ldots d$}
            {
        		$E = S[\zyl{m_{src}}\times b+p_1][j]$  \\
                \For{$shifter \leftarrow 1 \ldots 2$} 
                {
                    \uIf{the stored index equals $i$ and the stored fingerprint matches $f(src)$}
                    {
                        $H(temp)=s(temp)\times F+f(temp)$  \\
                        $q.push(H^{-1}(temp))$ \\ 
                        $checked.add(H^{-1}(temp))$ \\
                    }
                }
            }
        }
    }
\end{algorithm}

We define a queue $q$ and a checked list \zzzyl{$checked$} to implement BFS. 
After adding $A$ to $q$ and $checked$ (lines 1-2), we start to traverse $q$. The loop stops when $q$ is empty, and we return false. If the reachable path is found in advance, we return true.

In each loop, we first calculate the possible $r$ rows/columns for the extracted temporary starting vertex $src$ and $B$, respectively (lines 4-5). Then, we traverse the possible twin segments in turn to check whether the index pairs and the fingerprint pairs are both matched (lines 6-11). If so, we already find a reachable path from $A$ to $B$ and return true (lines 12-14); otherwise, we add all successors of the $src$ from the additional pool and the main matrix to $q$ and $checked$ (lines 15-24).
Note that we need to allocate storage space to store the initial address $s(v)$ of each node $v$.
And the function $H^{-1}$ receives a hash value and returns its corresponding vertex identifier.


The accuracy of the path queries without label restriction is ensured by $P^{|V|-d_a}$, where $P$ is calculated in Theorem \ref{edgecollision}, and $|V|$ and $d_a$ denote the total number and the average out-degree of vertices in $G$ respectively. The accuracy of the path queries with edge label restriction is bounded by $P^{|V|-d_a}\times (1-\frac{1}{c})^{l-1}$, where $l$ is the number of edge labels and $c$ is the hash range of them.
In addition, the time complexity of the queries is determined by the adopted path reachability algorithm.

\subsection{\zyl{Approximate} Subgraph Queries}
\zyl{LSketch preserves the graph connectivity, so the aggregated results of edges can actually represent the graph connectivity. 
Therefore, in order to improve the query efficiency, we 
use an approximate subgraph query to represent the subgraph query results with a certain degree of accuracy.}
Given a subgraph composed of a set of edges $G_{q}=\{e_1,...,e_n\}$, a subgraph query estimates the number of subgraph patterns that match. In our schema, an edge $e_i$ is expressed as $(x_i,l_{x_i},y_i,l_{y_i})$, which corresponds to its two vertices and related vertex labels.
Simply put, we only need to perform an edge query on each edge to get its matching weight, and then take the minimum of all weights as the number of matches.
The detailed algorithm is shown in Algorithm 6. Besides, the subgraph queries can be easily extended to the case with edge label constraints.

\begin{algorithm}[htbp]
	\caption{\zyl{Approximate}SubgraphQueries}
	\LinesNumbered
	\small{ \KwIn{a subgraph $G_{q}$ composed of a group of edges \zyl{$\{e_1,...,e_v\}$, where $e_i=(x_i,l_{x_i},y_i,l_{y_i})$    }}}
	\KwOut{the number of matches of the queried subgraph $G_{q}$}
   
    $res = MAX $ \\
    \For{$i \leftarrow 1 \ldots\zyl{v}$}
    {
        $w_i = AggregateEdgeQueries(x_i,l_{x_i},y_i,l_{y_i})$   \\
        \uIf{$w_i=0$}
        {
            return 0
        }
        \uIf{$w_i<res$}
        {
            $res = w_i$
        }
    }
    return $res$
\end{algorithm}

The accuracy of the \zyl{approximate} subgraph queries is similar to the edge queries, which is ensured by \zyl{$P^{v}$ and $P^{v}\times (1-\frac{1}{c})^{l-1}$} for cases without/with edge label restriction, where $l$ is the number of edge labels, $c$ is the hash range of them, and \zyl{$v$} denotes the size of the queried subgraph.
Also, the time complexity of subgraph queries is \zyl{$O(v)$}. 

\section{EXPERIMENTAL EVALUATION} \label{sec:experiments}
\subsection{Experimental Settings}

\noindent
\textbf{Datasets.} We use four real world datasets \yl{shown in Table \ref{table:windowsize}.}



\begin{table}[htbp] 
\vspace{-0.5em}
  \centering
  \fontsize{8}{8}\selectfont
  \caption{Statistics of the datasets}
    \begin{tabular}{|c|c|c|c|c|}
    \hline
  \textbf{Dataset} & \textbf{Phone} & \textbf{Road} & \textbf{Enron} & \textbf{com-FS} \\  
  \hline  
  \# of edges       & 60,765 & 870,757 & 2,064,442 & 1,806,067,135 \\ \hline
\# of vertex labels & 2      & --       & 11     & 20           \\ \hline
\# of edge labels & 4      & 6       & 35455     & 100           \\ \hline
window size       & 1 week & 1 day   & 1 week    & 1 day         \\ \hline
subwindow size    & 1 h    & 5 min   & 1 h       & 10 min        \\ 
  \hline  
    \end{tabular}
    \label{table:windowsize}
\end{table}

\noindent
\textbf{1) Phone data}: the MIT Reality dataset\yl{\footnote{https://crawdad.org/mit/reality/20050701}}.
The project recorded the communication, proximity, location, and activity information of 94 subjects from 2004 to 2005. Each item in the graph stream includes calling ID, receiving ID, call time, call type and call duration. 
We preprocess the graph stream, classify the vertices into two groups: research objects and others, and assign nine edge labels according to the type and duration of the call. The dataset has a total of 60,765 edges.
\textbf{2) Road data}: the real-time traffic speed map of the roads in Hong Kong\yl{\footnote{https://data.gov.hk/sc-data/dataset/hk-td-sm\_1-traffic-speed-map}}. 
We extract the IDs of the two endpoints, region, road type, road saturation and traffic speed of the road of each observation from the downloaded XML file, and model them as a graph stream according to the timestamps. Among them, we do not consider vertex labels and only assign six edge labels to describe the road situation according to its type and saturation. The dataset has a total of 870,757 edges.
\textbf{3) Email network data}: an email interaction dataset
with 2,064,442 edges\yl{\footnote{http://www.ahschulz.de/enron-email-data}}. We use the employee's position as the vertex label, resulting in a total of eleven types. The edge label of each item is the subject of the email, which is hashed and represented by a unique number.
\textbf{4) Social network data}: the Friendster social network from SNAP \cite{snap}. The dataset is a static network containing 1,806,067,135 edges. We generate a set of vertex labels with size 20 and a set of edge labels with size 100, and randomly allocate labels and timestamps to each item. The final size of the dataset is 32.9 GB. \cy{We use this semi-synthetic dataset to examine the scalability of the proposed method.}

\noindent
\textbf{Competitive Methods.}
\yl{
We choose GSS \cite{gou2022graph} and LGS \cite{Song2019LabeledGS} as our competitive methods. 
Although GSS cannot handle heterogeneous graphs or differentiate items arriving at different timestamps, it is the latest method that most relevant to ours with high accuracy. 
LGS is capable of processing labels and timestamps and is the state-of-the-art method working for heterogeneous graph streams.
We experimentally evaluate the efficiency and accuracy of LSketch against GSS and LGS. 
Among them, we compare LSketch with \cy{both of }GSS \cy{and} LGS in queries without label restriction\cy{s,} and with LGS \cy{only} in queries with label restriction\cy{s}.
}

\begin{figure*}[ht]
\centering
\setlength{\abovecaptionskip}{0.1cm} 
\setlength{\belowcaptionskip}{-0.23cm}   
\subfigure[{Phone} \label{fig:width_p}]{
\begin{minipage}[htbp]{0.3\textwidth}
\includegraphics[width=1\textwidth]{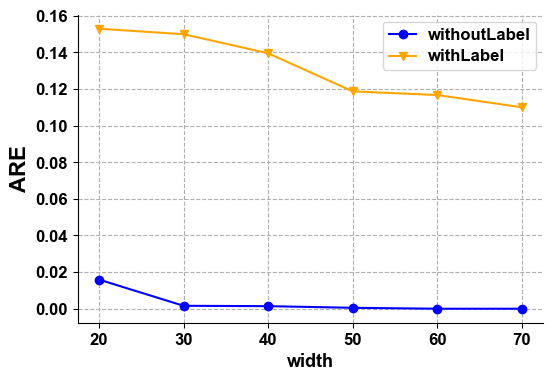}
\end{minipage}
}
\subfigure[{Road} \label{fig:width_r}]{
\begin{minipage}[htbp]{0.3\textwidth}
\includegraphics[width=1\textwidth]{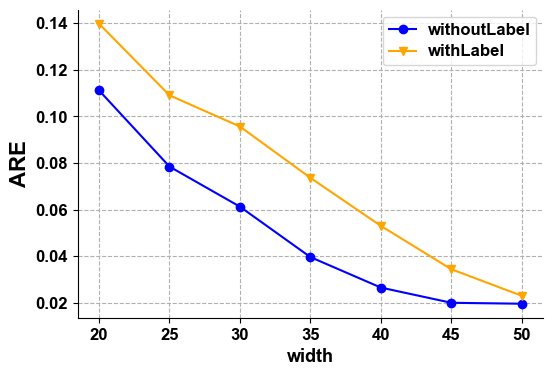}
\end{minipage}
}
\subfigure[Enron  \label{fig:width_e}]{
\begin{minipage}[htbp]{0.3\textwidth}
\includegraphics[width=1\textwidth]{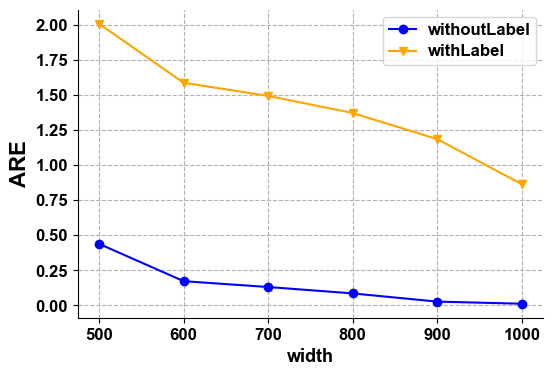}
\end{minipage}
}
\caption{Vertex queries varying $d$ of the matrix}
\end{figure*}

\noindent
\textbf{\cy{Parameter }
Settings.}
For a fair comparison, we set the item weight of all datasets to 1 to facilitate GSS and LGS to process these graph streams, and do not consider labels and timestamps in GSS. 
As for LGS, we use \yl{6} \cy{copies of }graph sketches to improve its accuracy.
\yl{For the three methods, $d$ is set to be the same, so LGS will use \cy{six }
times the storage space to compare with GSS and LSketch
.
Specifically, we will discuss how to set the size of $d$ for different datasets later.
\cy{In addition, }
the size of $b$ \cy{(for LGS and LSketch) }
are set according to $d$ and the number of vertex labels of the datasets, \cy{and }the fingerprint sizes differ from 8-16 bits. We set the length of the candidate list to 16 and the sampled list length $s$ to 16.}
According to the real applications of the graph streams and its overall time span, the window size and the subwindow size of each dataset is shown in Table \ref{table:windowsize}.

\noindent
\textbf{Metrics.}
We use \textbf{average relative error (ARE)} to describe the accuracy of vertex queries, edge queries and subgraph queries. Denote the estimation and the true results of a query as $res_s$ and $res$, respectively, and the relative error is \cy{computed } 
as $\frac{res_s-res}{res}$ \cy{since $res_s$ will be no less than $res$ due to the possible hash collisions}. For each type of query, we randomly select the query set and obtain the \cy{ARE from 500 repeated runs. }
For path queries, the error only occurs when the real result is false and the sketch returns true. Therefore, we use \textbf{accuracy} (the number of false positives/the number of queries) to measure \cy{the }
performance.

\noindent
\textbf{Setup.} We perform all our experiments on a desktop with 16GB memory, Intel Core i7 processor and 3.40GHz frequency. All algorithms including LSketch, GSS and LGS are implemented in C++.

\subsection{\cy{Evaluation on matrix width $d$}
}
\cy{In this section, we will discuss how to set the most important parameter, $d$--the width of the matrix, by showing the evaluation results varying $d$. }After determining $d$, we can simply calculate the size of the block based on the number of vertex labels in the dataset (when using the uniform blocking strategy). \cy{Time related prameters, such as subwindow and window size, are set according to actual application requirements. }
For other parameters, we follow the suggestions of GSS.

In general, the matrix is the main part that holds the edges and its capacity is $d*d*2$. Therefore, this capacity should be comparable to the number of distinct edges of the dataset. Taking the Phone dataset as an example, there are 4952 distinct edges, so the most suitable $d$ should be around 50. However, since our method involves the block division caused by the encoding of vertex labels, it is usually necessary to set a larger $d$ in order to obtain a higher accuracy. Following the above guidance, we set different $d$ to generate sketches and measure their performance using the results of vertex queries without/with edge label constraints, which is shown in Figure \ref{fig:width_p}. Note that in order to show the performance difference of sketches more clearly, we set \cy{a small fingerprint length. }
It can be seen that with the growth of $d$, the ARE of the query decreases. When $d$ reaches 60, there are no errors.

Similarly, it is calculated that the recommended \cy{$d$ }
for Road and Enron is 40 and 600, respectively. The experimental results are shown in Figure \ref{fig:width_r} and Figure \ref{fig:width_e}. The \cy{variation }
trend of ARE is in line with our inference, and the error around the recommended \cy{width }
is acceptable. 
\zyl{For GSS and LGS, the authors also discuss the matrix width $d$ in their papers respectively, and their experimental results are consistent with our conclusions. For all sketches, the ARE of the queries decreases as the matrix width increases and eventually reaches saturation.}

\cy{In practice, for streaming enabled scenarios, $d$ can be set according to the edge incoming rate within one window.}
\yl{In subsequent experiments, we use the recommended $d$ for each dataset to construct sketches.}

\subsection{\cy{System Throughput}
}
In this section, we \cy{examine the scalability of the sketches by investigating their system throughput.}
Since GSS \cy{is }
not \cy{able }
to handle timestamps, we use LGS and LSketch without sliding windows \cy{to show the overall comparison results}.
\yl{Table \ref{table:inserttime1} shows the average time \cy{each method}  takes to insert an item and the overall insertion time for each dataset}. It can be seen that on the first three datasets, the average insertion time of all methods is within the $\mu s$ level.
Among them, LSketch \cy{preserves much more information than GSS, and achieves much higher accuracy than LGS (shown in Section \ref{query evaluation}), thus may take }
a little longer loading time.
On very large datasets, all three methods take the \cy{$ms$} level average time to complete \cy{one }insertion.
This is because even if only around 1\% of the edges are stored in the additional pool, inserting and updating these edges is still costly due to the inefficiency of adjacency list. 
     

     

\begin{table}[htbp]
  \centering
  \fontsize{8}{8}\selectfont
  \caption{\yl{Average insertion time (us\cy{/edge}) and total insertion time (ms)}}
    \begin{tabular}{c|c|rrrr}
    \toprule
    \multirow{2}[4]{*}{\textbf{Time}} & \multirow{2}[4]{*}{\textbf{Methods}} & \multicolumn{4}{c}{\textbf{Datasets}} \\
\cmidrule{3-6}          &       & \multicolumn{1}{c|}{\textbf{Phone}} & \multicolumn{1}{c|}{\textbf{Road}} & \multicolumn{1}{c|}{\textbf{Enron}} & \multicolumn{1}{c}{\textbf{com-FS}} \\
    \midrule
    \multirow{3}[6]{*}{\makecell[c]{Average \\ (us/edge)}} & GSS   &  1.22 & 1.01 & 2.77  & 5.89 ms  \\
\cmidrule{2-6}          & LGS   & 1.56  & 1.22  & 11.43  & 6.77 ms  \\
\cmidrule{2-6}          & LSketch & 2.73  & 1.83  & 7.27 & 8.82 ms  \\
    \midrule
    \multirow{3}[6]{*}{\makecell[c]{Total \\ (ms)}} & GSS   &  74.3 & 875.3 &  5724.5 & 3210000 s \\
\cmidrule{2-6}          & LGS   & 94.9  & 1062.6  & 23593.5  & 3670000 s \\
\cmidrule{2-6}          & LSketch & 166.1  & 1591.2  & 14998.5 & 4780000 s  \\
    \bottomrule
    \end{tabular}%
  \label{table:inserttime1}%
\end{table}%

     

\begin{table}[htbp]
  \centering
  \fontsize{8}{8}\selectfont
  \caption{Insertion time (ms \cy{per edge/overall}) \cy{with sliding windows}}
    \begin{tabular}{c|c|r|r}
    \toprule
    \multirow{2}[4]{*}{\textbf{Time}} & \multirow{2}[4]{*}{\textbf{Methods}} & \multicolumn{2}{c}{\textbf{Datasets}} \\
\cmidrule{3-4}          &       & \multicolumn{1}{c|}{\textbf{Phone}} & \multicolumn{1}{c}{\textbf{Road}} \\
    \midrule
    \multirow{2}[4]{*}{\makecell[c]{Average \\ (ms/edge)}} & LGS   &   1.36    &  0.03 \\
\cmidrule{2-4}          & LSketch &   0.43    &  2.45 us\\
    \midrule
    \multirow{2}[4]{*}{\makecell[c]{Total \\ (ms)}} & LGS   &   82739.40    &   23468.90 \\
\cmidrule{2-4}          & LSketch &  26429.70     & 2137.80 \\
    \bottomrule
    \end{tabular}%
  \label{table:inserttime3}%
\end{table}%

Furthermore, to compare LSketch and LGS in more detail, we also performed the same experiments on Phone and Road \cy{datasets with sliding windows}. The results in Table \ref{table:inserttime3} show that 
\cy{LSketch }
is better suited for handling heterogeneous graph streams compared with LGS \cy{in aspect of time efficiency}.
The average insertion time of LSketch is an order of magnitude smaller than that of LGS\cy{.}

\zyl{Our approach is a universal storage structure that supports a wide range of graph queries. Therefore, it has more powerful features compared to those structures for specific graph queries optimization. In particular, with the introduction of sliding windows, we are able to support various types of structured queries under time region constraints, which is of great practical significance.}

\subsection{\cy{Evaluation on Query Answering}} \label{query evaluation}

\noindent
\textbf{1) Time Efficiency.}
Sketches show a significant advantage in supporting queries in terms of time efficiency, which is several orders of magnitude times faster than querying on raw data. The results of vertex queries and edge queries are shown in \yl{Table \ref{table:queriestime}} for an illustration.
During the query process, LSketch needs to further deal with two vertex labels, and thus the time consumption is slightly higher than that of GSS. However, such disparity is almost negligible, since all queries can be finished within the $\mu s$ level when using sketches. The time consumption of other types of queries is similar, and thus we omit \cy{them }
for clarity.

     

     

\begin{table}[htbp]
  \centering
  \fontsize{8}{8}\selectfont
  \caption{\yl{The response time of the vertex queries and edge queries }}
    \begin{tabular}{c|c|rrr}
    \toprule
    \multirow{2}[4]{*}{\textbf{Queries}} & \multirow{2}[4]{*}{\textbf{Methods}} & \multicolumn{3}{c}{\textbf{Datasets}} \\
\cmidrule{3-5}          &       & \multicolumn{1}{c|}{\textbf{Phone}} & \multicolumn{1}{c|}{\textbf{Road}} & \multicolumn{1}{c}{\textbf{Enron}} \\
    \midrule
    \multirow{4}[5]{*}{\makecell[c]{Vertex \\ queries}} & Raw data   & 73.90 ms & 934.58 ms & 8388.80 ms  \\
\cmidrule{2-5}          & GSS   &  4.18 us & 5.47 us & 72.87 us  \\
\cmidrule{2-5}          & LGS & 2.21 us & 3.98 us & 68.88 us  \\
\cmidrule{2-5}          & LSketch & 4.67 us & 5.13 us & 76.87 us  \\
    \midrule
    \multirow{4}[5]{*}{\makecell[c]{Edge \\ queries}} & Raw data     & 332.59 ms & 3611.56 ms & 32208.23 ms \\
\cmidrule{2-5}          & GSS &  1.07 us & 1.19 us & 6.30 us \\
\cmidrule{2-5}          &  LGS & 1.60 us & 1.38 us & 6.43 us \\
\cmidrule{2-5}     &  LSketch & 3.04 us & 1.05 us & 10.77 us    \\
    \bottomrule
    \end{tabular}%
  \label{table:queriestime}%
\end{table}%

\noindent
\textbf{2) The Accuracy of Queries When Ignoring Timestamps. }

\textbf{Vertex queries.}
Given the queried vertices, we perform out-degree queries on three sketches of all datasets, and the results are shown in Figure \ref{fig:node}.
It can be seen that our method is much better than LGS on all datasets. Since we preserve more information of the graph streams with no extra storage space\cy{, the }
accuracy is slightly worse than GSS on the Phone dataset. On other datasets, LSketch achieves the same accuracy as GSS, and this demonstrates that our method is efficient and accurate in vertex label preservation.

We also calculate the error rate guarantees of all datasets. The probability $P$ of the Phone dataset is $0.999352$, and the error rate of vertex queries is guaranteed to be less than $0.219$, which is much higher than the actual running value. On other datasets, the results are similar. The probability that edges do not collide is around $0.9$, and the calculated error rate is guaranteed to be very small.

In addition, to demonstrate the superiority of LSketch in maintaining heterogeneous information, we perform the vertex queries with edge label restrictions, which is shown in Figure \ref{fig:node_label}. 
We only show the results of LSketch and LGS, since GSS does not preserve label information, and hence is not able to support label \cy{constrained }
queries.
The performance is good because we set a larger $c$ for all datasets{, where }\yl{$c$ is the length of the predefined list of prime numbers.}
The larger the $c$, the smaller the probability of edge label collisions, contributing to a lower error rate. We can see that LSketch outperforms LGS quite a bit under the same parameter settings. Moreover, there's no need for LSketch to use multiple sketches to improve query accuracy, demonstraing LSketch's abilities of information preservation, space saving, and query answering efficiency.

\begin{figure*}[htbp]
\centering
\setlength{\belowcaptionskip}{-0.23cm}   
\subfigure[{Vertex queries} \label{fig:node}]{
\begin{minipage}[htbp]{0.22\textwidth}
\includegraphics[width=\textwidth]{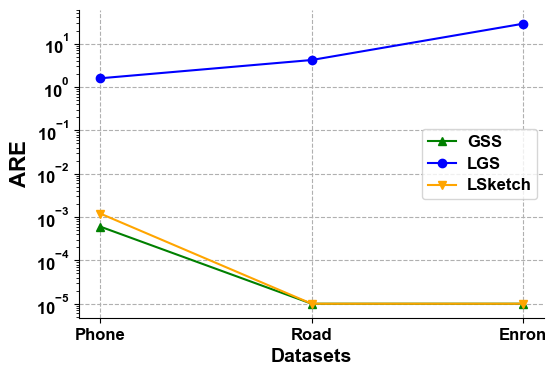}
\end{minipage}
}
\subfigure[{Vertex queries (lc)} \label{fig:node_label}]{
\begin{minipage}[htbp]{0.22\textwidth}
\includegraphics[width=\textwidth]{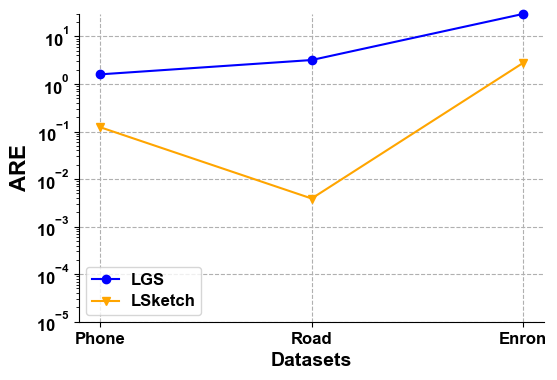}
\end{minipage}
}
\subfigure[Edge queries \label{fig:edge}]{
\begin{minipage}[htbp]{0.22\textwidth}
\includegraphics[width=\textwidth]{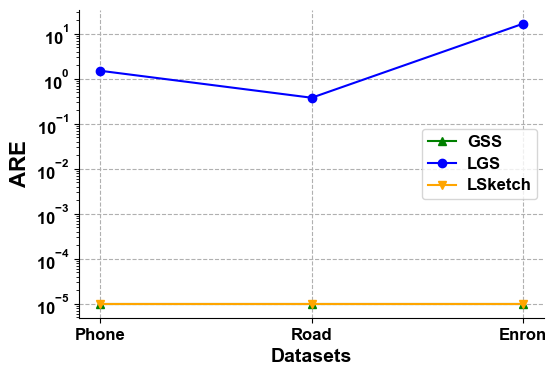}
\end{minipage}
}
\subfigure[Edge queries (lc) \label{fig:edge_label}]{
\begin{minipage}[htbp]{0.22\textwidth}
\includegraphics[width=\textwidth]{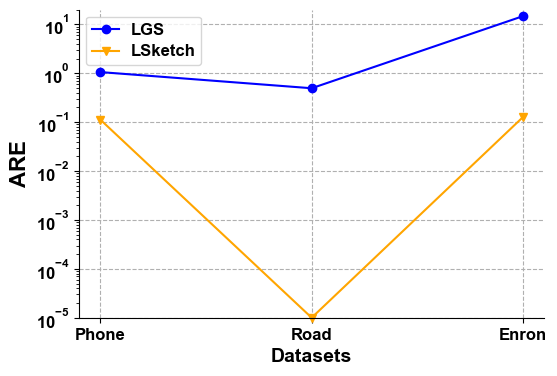}
\end{minipage}
}
\subfigure[{Path queries} \label{fig:path}]{
\begin{minipage}[htbp]{0.22\textwidth}
\includegraphics[width=\textwidth]{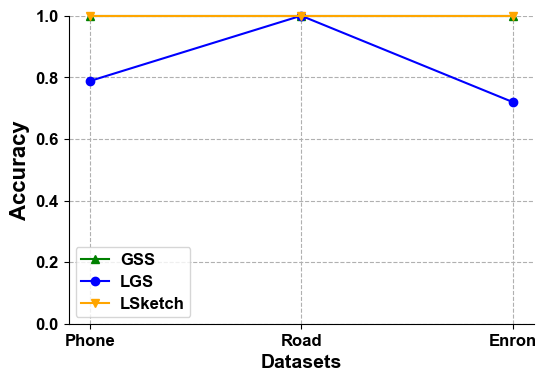}
\end{minipage}
}
\subfigure[{Path queries (lc)} \label{fig:path_label}]{
\begin{minipage}[htbp]{0.22\textwidth}
\includegraphics[width=\textwidth]{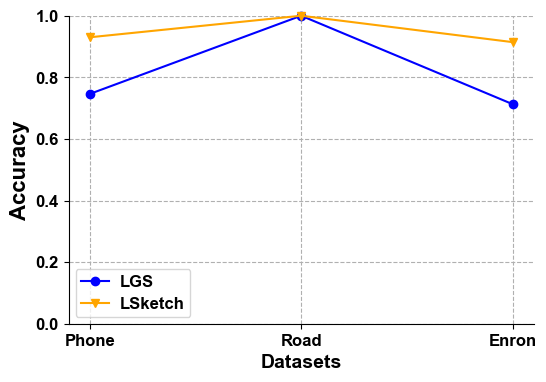}
\end{minipage}
}
\subfigure[Subgraph queries  \label{fig:sub}]{
\begin{minipage}[htbp]{0.22\textwidth}
\includegraphics[width=\textwidth]{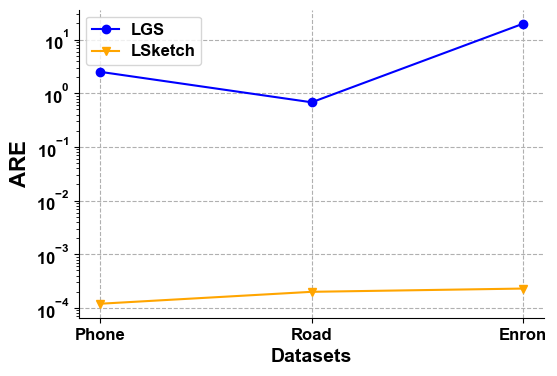}
\end{minipage}
}
\subfigure[\tiny Subgraph queries (lc) \label{fig:sub_label}]{
\begin{minipage}[htbp]{0.22\textwidth}
\includegraphics[width=\textwidth]{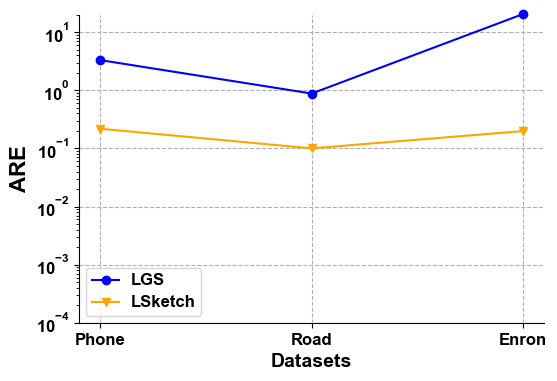}
\end{minipage}
}
\caption{The performances of the queries for three sketches without sliding windows}
\end{figure*}

\textbf{Edge queries.}
Next, we show the results of edge queries over the three sketches.
From Figure \ref{fig:edge}, we can see that the three methods all perform well on edge queries. Particularly, LSketch and GSS are both nearly $100\%$ accurate on all datasets. LGS does not have the ability to distinguish different items stored in the same location, \cy{thus }
its error rate is slightly higher.
The error rate guarantee equals $1-P$, which is the probability that an edge collides with other edges, and can be maintained below $0.1$ on all datasets.
The results of edge queries with edge label restrictions are similar and are shown in Figure \ref{fig:edge_label}.

\textbf{Path queries.}
Now we evaluate the performance of LSketch in supporting reachability queries in Figures \ref{fig:path} and \ref{fig:path_label}. 
The accuracy of the Road dataset is always 1 since it is a connected graph. For other datasets, the accuracy of LSketch is comparable to GSS and is much better than LGS. 
The accuracy guarantee of the path queries is the same as that of vertex queries as analyzed \cy{above.}

\textbf{\zyl{Approximate} Subgraph queries.}
Since the basic version of GSS does not support subgraph queries, we only compare the results with LGS.
The approximate matching query that we implement is a repeated execution of the edge query; hence the results of the subgraph queries are quite similar to those of the edge queries, as shown in Figures \ref{fig:sub} and \ref{fig:sub_label}.

\noindent
\textbf{3) The Accuracy of Queries After Introducing the Sliding Window. }
Figure \ref{fig:window_exper} shows the results of vertex queries and edge queries on LSketch and LGS, where 'lc' means that the query is under label constraints. It can be seen that with the introduction of the sliding window, the ARE of LSketch is further reduced, far superior to LGS. Other results are omitted due to space constraints.

\begin{figure}[htbp]
\centering
\setlength{\abovecaptionskip}{0.1cm} 
\setlength{\belowcaptionskip}{-0.23cm}   
\subfigure[{Vertex queries} \label{fig:nodewindows}]{
\begin{minipage}[htbp]{0.225\textwidth}
\includegraphics[width=\textwidth]{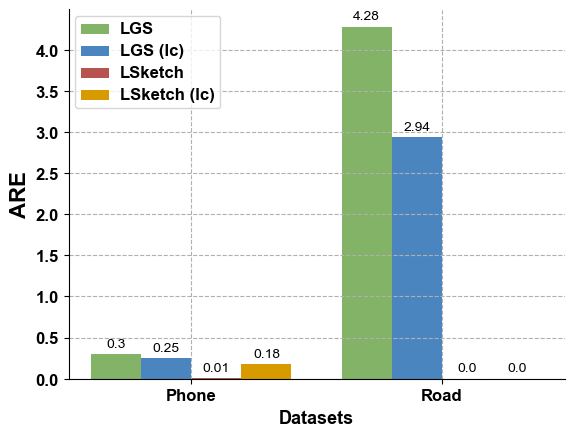}
\end{minipage}
}
\subfigure[{Edge queries} \label{fig:edgewindows}]{
\begin{minipage}[htbp]{0.225\textwidth}
\includegraphics[width=\textwidth]{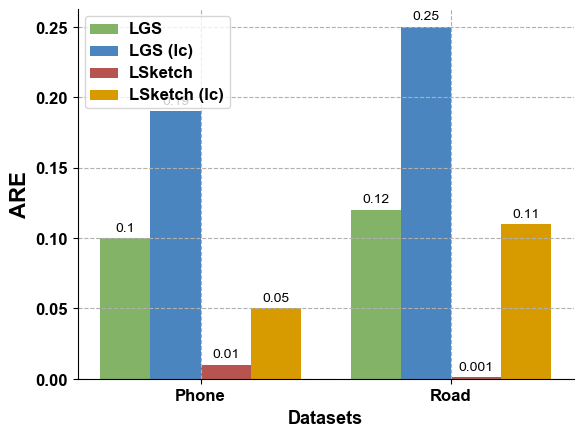}
\end{minipage}
}
\caption{The performances for three sketches with sliding windows}
\label{fig:window_exper}
\end{figure}

\section{Conclusion}
In this paper, we propose a novel structure LSketch for graph stream summarization. It employs a sliding window mechanism and works for heterogeneous graph streams, which is more in line with the needs of real applications.
LSketch only takes a sub-linear storage space and $O(1)$ update cost. It preserves the underlying structure and the label information of graph streams, enabling it to support multiple types of structure based queries.
The experimental results show that LSketch enjoys a great improvement in query accuracy and response speed compared to LGS.
In addition, compared to the state-of-the-art method GSS, which works for homogeneous graph streams, our proposed method further maintains the storage of labels and timestamps with a slightly more time cost.
The above experimental results and theoretical analyses fully demonstrate the superiority of LSketch.

\ifCLASSOPTIONcompsoc
  \section*{Acknowledgments}
\else
  \section*{Acknowledgment}
\fi

This work was supported in part by NSFC under the grants 62172237, 61772289, U1836109, U1936206, U1936105 and 62077031; NSF grant IIS-1633271, and New England Transportation Consortium project 20-2.

\ifCLASSOPTIONcaptionsoff
  \newpage
\fi



\bibliographystyle{IEEEtran}
\bibliography{ref}
%



%

\begin{IEEEbiography}[{\includegraphics[width=1in,height=1.25in,clip,keepaspectratio]{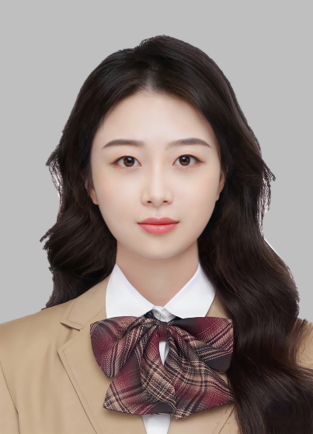}}]{Yiling Zeng}
received her BS degree from Nankai University, China in 2020. She is currently a master candidate at Nankai University. Her research interests include graph summarization and graph sketches on graph streams.
\end{IEEEbiography}
\vspace{-1.2cm}

\begin{IEEEbiography}[{\includegraphics[width=1in,height=1.25in,clip,keepaspectratio]{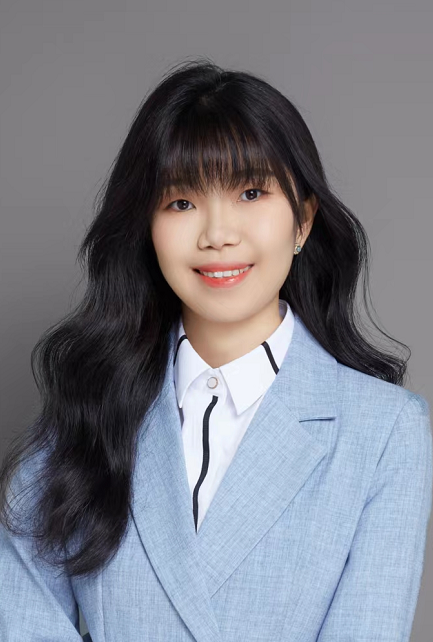}}]{Chunyao Song}
is an associate professor in Department of Computer Science at  Nankai University. Her research interests include graph data processing and  analysis, as well as streaming data processing and analysis.
\end{IEEEbiography}
\vspace{-1.2cm}

\begin{IEEEbiography}[{\includegraphics[width=1in,height=1.25in,clip,keepaspectratio]{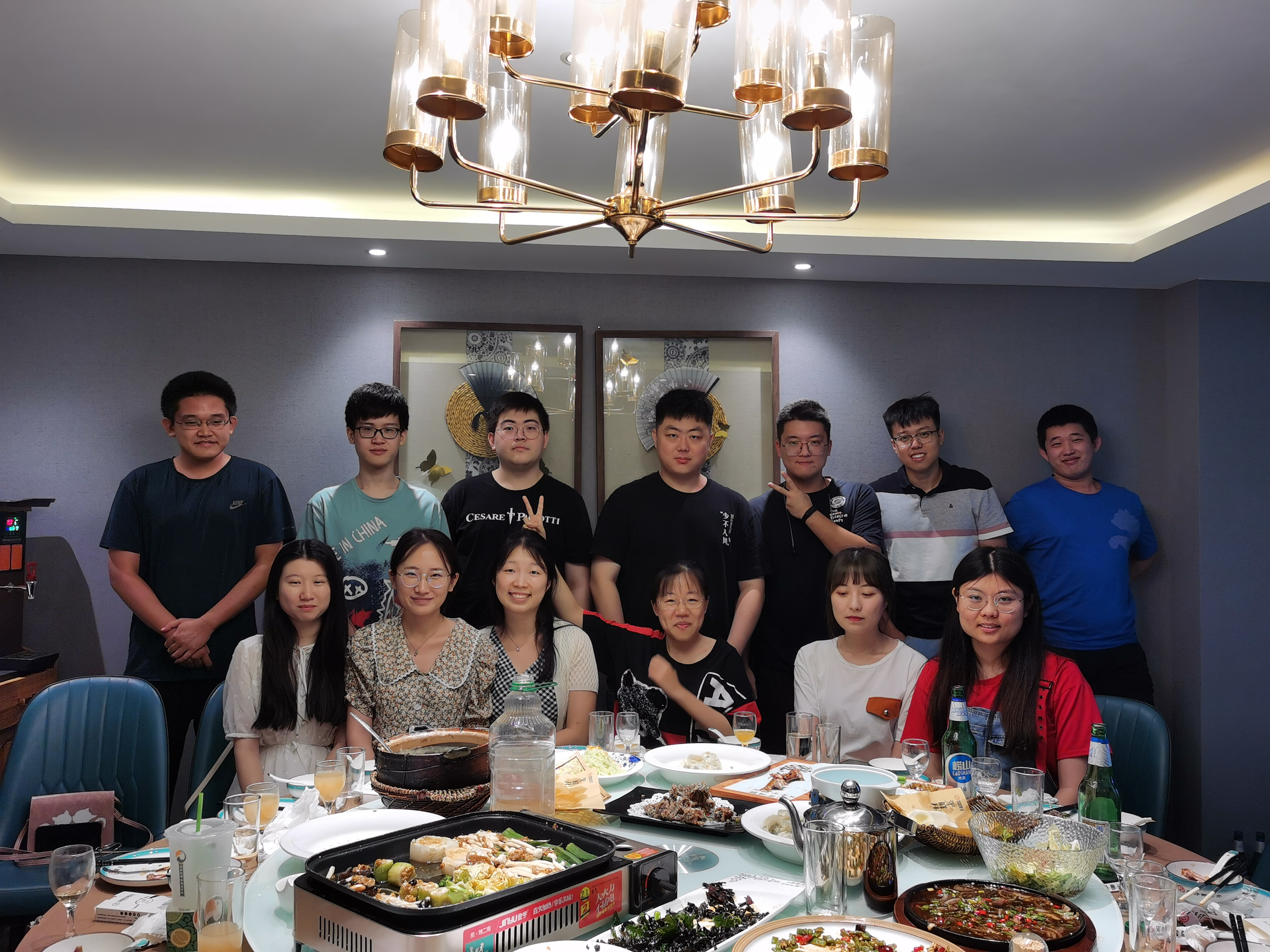}}]{Yuhan Li} received his BS degree from Northeast Forestry University, China in 2020. He is currently a master candidate at Nankai University. His research interests include knowledge graph, entity linking and data mining.
\end{IEEEbiography}
\vspace{-1.2cm}

\begin{IEEEbiography}[{\includegraphics[width=1in,height=1.25in,clip,keepaspectratio]{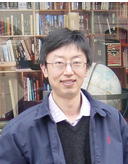}}]{Tingjian Ge}
is a professor in Department of Computer Science at the University of Massachusetts Lowell. His research interests include data streams, graphs and graph streams, noisy and uncertain data, biomedical data analysis, and data security and privacy.
\end{IEEEbiography}


\vfill


\end{document}